\documentclass[11pt]{article}

\usepackage[utf8]{inputenc}
\usepackage{amsmath}
\usepackage{amsfonts}
\usepackage{amssymb}
\usepackage{amsthm}
\usepackage{bbm}
\usepackage{enumitem}
\usepackage{geometry}
\usepackage{graphicx}
\usepackage{import}
\usepackage{listings}
\usepackage{multicol}
\usepackage{soul}
\usepackage{standalone}
\usepackage{titlesec}
\usepackage{tikz}

\newtheorem{Definition}{Definition}
\newtheorem{Theorem}{Theorem}
\newtheorem{Lemma}{Lemma}
\newtheorem{Observation}{Observation}
\newtheorem{Corollary}{Corollary}

\newcommand{\rubenAugust}[1]{\textcolor{teal}{#1}}

\newcommand{\MAF}{{\rm uMAF}}

\usetikzlibrary{decorations.pathreplacing}
\newcommand{\vertex}{\node[vertex]}
\tikzstyle{vertex}=[draw, shape=circle, minimum size=0.5em, inner sep=1, fill]

\title{Agreement forests of caterpillar trees: complexity, kernelization and branching}
\author{S. Kelk and R. Meuwese\footnote{Both authors are at the Department of Advanced Computing Sciences (DACS), Maastricht University, The Netherlands. Email: \texttt{steven.kelk@maastrichtuniversity.nl, r.meuwese@maastrichtuniversity.nl.}}}
\date{}

\begin{document}

\maketitle

\begin{abstract}
Given a set $X$ of species, a phylogenetic tree is an unrooted binary tree whose leaves are bijectively labelled by $X$. Such trees can be used to show the way species evolve over time.
One way of understanding how topologically different two phylogenetic trees are, is to construct a minimum-size agreement forest: a partition of $X$ into the smallest number of blocks, such that the blocks induce homeomorphic, non-overlapping subtrees in both trees.  This comparison yields insight into commonalities and differences in the evolution of $X$ across the two trees.
Computing a smallest agreement forest is NP-hard \cite{HeinJWZ96}. 
In this work we study the problem on caterpillars, which are path-like phylogenetic trees. 
We will demonstrate that, even if we restrict the input to this highly restricted
subclass, the problem remains NP-hard and is in fact APX-hard.
Furthermore we show that for caterpillars two standard reduction rules well known in the literature yield a tight kernel of size at most $7k$, compared to $15k$ for general trees  \cite{KelkL18}.
Finally we demonstrate that we can determine if two caterpillars have an agreement forest with at most $k$ blocks in $O^*(2.49^k)$ time, compared to $O^*(3^k)$ for general trees \cite{ChenFS15}, where $O^*(.)$ suppresses polynomial factors.
\end{abstract}

\section{Introduction}
In biology phylogenetic trees are commonplace. These are leaf-labelled trees which show how
the entities $X$ at the leaves - most commonly, but not exclusively, species - evolve over time \cite{SempleS03}. Unlabelled interior nodes of the trees represent points in time at which hypothetical common ancestors diversified into sub-lineages. Such trees are typically built from data that carries evolutionary signal, such as DNA sequences. However, in the real world there is no unique mapping from DNA sequences to the ``true'' tree; it depends on the quality of the available data, underlying biological phenomena and multiple model assumptions. Hence, for given data carrying evolutionary signal for a set of species $X$ it might be possible to generate multiple distinct, but equally plausible, trees, for $X$. A significant part of the phylogenetics literature is therefore dedicated to understanding when and why trees differ, a phenomenon known as \emph{incongruence} or \emph{discordance} \cite{degnan2009gene}.

One model for summarizing the topological difference of two trees is the \emph{agreement forest}.  Freely translated, an agreement forest is a summary of the topological building blocks common to both trees. More formally, an agreement forest is a partition of the leaf set $X$ such that each block induces the same topology in both trees and, within each tree, the induced subtrees do not overlap. More details will be provided in section \ref{Section Pre}.

A partitition of $X$ into singletons is vacuously a valid agreement forest, but this does not provide any insight. Rather, in the spirit of parsimony we wish to have an agreement forest with a minimum number of building blocks; this is called a \emph{maximum} agreement forest (MAF), so called because it is an agreement forest that maximizes the agreement between the two trees. Maximum agreement forests have been studied extensively in recent years, we refer to \cite{ChenSW16,BulteauW19,KelkL19} for overviews. Here we focus on the situation when the input consists of two unrooted, binary trees, writing uMAF to distinguish our problem from the rooted variant. Unfortunately, even in this limited setting finding the number of blocks in an uMAF, $d_\MAF(T,T')$, is an NP-hard problem \cite{HeinJWZ96}. Nevertheless, it is appealing to try to solve the problem in practice, not least because of its close relationship to several other measurements used to compare two phylogenetic trees. In particular, it is closely related to the Tree Bisection and Reconnection (TBR) distance,  $d_{TBR}(T, T')$, which (informally) counts the number of times a subtree has to be detached and reconnected to transform $T$ into $T'$; specifically, we have $d_{TBR}(T, T') = d_\MAF(T,T')-1$ \cite{AllenS01}. Distances such as TBR help us to understand the underlying connectivity of tree space \cite{john2017shape}. Interestingly, $d_{TBR}(T, T')$ is in turn exactly equal to the \emph{hybridization number} of $T$ and $T'$, which is the smallest value of $|E|-(|V|-1)$ ranging over all phylogenetic networks  $G=(V,E)$, generalizations of phylogenetic trees to graphs, that topologically simultaneously embed $T$ and $T'$ \cite{IerselKSSB18}. This graph-theoretic characterization has been central to recent parameterized complexity results for
$d_{TBR}$ and $d_\MAF(T,T')$ \cite{KelkL19}.

In this article we aim to develop a more fine-grained understanding of what makes computation of $d_\MAF$ challenging. We do this by restricting our attention to the problem when the input consists of two \emph{caterpillars}; these are path-like phylogenetic trees. An example is given in Figure \ref{Fig example of caterpillar}.

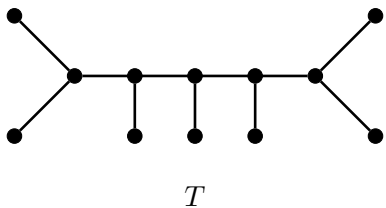
\begin{figure}[!h]
\begin{center}
\begin{tikzpicture}[scale = 0.8]

    \node (F1) at (3,-2) {$T$};
    \vertex  (a) at (0,1) {};
    \vertex  (b) at (0,-1) {};
    \vertex (ab) at (1,0) {};
    \vertex (cd) at (5,0) {};
    \vertex  (c) at (6,1) {};
    \vertex  (d) at (6,-1) {};
    
    \vertex (2x) at (2,0){};    
    \vertex (2) at (2,-1){};
    \vertex (3x) at (3,0){};    
    \vertex (3) at (3,-1){};
    \vertex (4x) at (4,0){};    
    \vertex (4) at (4,-1){};

    \draw [line width = 1pt]
    (ab) edge (cd)
    (ab) edge (a)
    (ab) edge (b)
    (cd) edge (c)
    (cd) edge (d)
    (2x) edge (2)
    (3x) edge (3)
    (4x) edge (4);
\end{tikzpicture}
\caption{Example of a caterpillar tree $T$.}
\label{Fig example of caterpillar}
\end{center}
\end{figure}

We prove several results.
In section \ref{Section NP Hard} we will prove that computing $d_\MAF$ for two caterpillars is NP-hard. In section \ref{Section APX Hard} we extend this result to APX-hardness, thus excluding the existence of a polynomial-time approximation scheme for computation of $d_\MAF$, unless P=NP. We note that the hardness is automatically inherited by the computation of $d_\MAF$ on general
trees. This is relevant because the APX-hardness of the general problem was stated, but not proven, in \cite{HeinJWZ96}. Our result thus closes this gap in the literature. In section \ref{Section Kernel} we will prove that for two caterpillars there is a tight $7k$ kernel, using just two reduction rules, where $k$ is equal to $d_\MAF$. Specifically: when applied exhaustively to two caterpillars, the well-known subtree and reduction rules yield a smaller pair of caterpillars that have at most $7k$ leaves. The same two reduction rules yield a tight $15k$ kernel on general trees \cite{KelkL18}. There are also smaller kernels of size $11k$ and $9k$ for general trees \cite{KelkL19,kelk2022deep}\footnote{In the articles \cite{KelkL18,KelkL19,kelk2022deep} $k$ refers to $d_{TBR}(T, T')$, not $d_\MAF(T,T')$, but given that the quantities only
differ by 1 this does not affect the multiplicative term in kernel sizes. Only the additive terms will differ by a constant.} but in order to obtain those smaller kernels far more complex reduction rules and analysis are needed. It is interesting that for caterpillars the subtree and chain reduction already bring us substantially below the smallest kernel for general trees. Next, in section \ref{Section Branching Alg} we will give a branching algorithm that is faster than, at the time of writing, the best branching algorithm by Chen et al. \cite{ChenFS15}. The algorithm of Chen et al. runs in time $O^*(3^k)$, on general trees, where the * suppresses polynomial factors. In contrast, our algorithm for two caterpillars runs in time $O^*(2.49^k)$.
Finally, in section \ref{sec:discussion} we conclude with a number of discussion points and open problems.

\section{Preliminaries}\label{Section Pre}
For general background on mathematical phylogenetics we refer to \cite{SempleS03, DressHKMS12}. An {\it unrooted binary phylogenetic $X$-tree} is an undirected tree $T =(V(T),E(T))$ where every internal vertex has degree 3 and whose leaves are bijectively labelled by a set $X$, where $X$ is often called the set of \emph{taxa} (representing the contemporary species, for example). We use $n$ to denote $|X|$ and often simply write phylogenetic tree when it is clear from the context that we are talking about an unrooted binary phylogenetic $X$-tree. Two phylogenetic trees $T, T'$ on $X$ are considered equal if there is an isomorphism between them that is the identity mapping on $X$ i.e. is label-preserving. A \emph{cherry} of a tree $T$ on $X$ is a pair of distinct taxa $x, y \in X$ which have a common parent in $T$. A tree $T$ on $|X| \geq 4$ taxa is a \emph{caterpillar} if it has exactly two cherries. For convenience we define all trees on $|X| \leq 3$ taxa to be caterpillars, too. An equivalent definition is that a tree $T$ is a caterpillar if, after deleting all taxa, the resulting tree is a path.

Let $T$ be a tree on $X$. For $X' \subseteq X$ we write $T[X']$ to denote the minimal subtree of $T$ spanning $X'$, and write $T|X'$ to denote the phylogenetic tree obtained from $T[X']$ by suppressing nodes of degree 2.

Let $T$ and $T'$ be two phylogenetic trees  on $X$. Let $F = \{ B_1,B_2,\ldots,B_k\}$ be a partition of $X$, where each block $B_i$ with $i\in\{1,2,\ldots,k\}$ is  referred to as a \emph{component} of $F$. We say that $F$ is an \emph{agreement forest} for $T$ and $T'$ if the following conditions hold.
\begin{enumerate}
\item [(1)] For each $i\in\{1,2,\ldots,k\}$, we have $T|B_i = T'|B_i$.
\item [(2)] For each pair $i,j\in\{1,2,\ldots,k\}$ with $i \neq j$, we have that $T[B_i]$ and $T[B_j]$ are vertex-disjoint in $T$, and $T'[B_i]$ and $T'[B_j]$ are vertex-disjoint in $T'$.
\end{enumerate}

Let $F=\{B_1,B_2,\ldots,B_k\}$ be an agreement forest for $T$ and $T'$. The \emph{size} of $F$ is simply its number of components, $k$. Moreover, an agreement forest with the minimum number of components (over all agreement forests for $T$ and $T'$) is called a \emph{maximum agreement forest (MAF)} for $T$ and $T'$. The number of components of a maximum agreement forest for $T$ and $T'$ is denoted by $d_\MAF(T,T')$. The \textsc{Unrooted Maximum Agreement Forest (uMAF)} problem is to compute $d_\MAF(T,T')$. It is NP-hard \cite{HeinJWZ96}, but permits a polynomial-time 3-approximation \cite{WhiddenBZ13,whidden2009unifying}.\\

\textbf{Subtrees and chains.} Let $T$ be a phylogenetic tree on $X$.  We say that a subtree of $T$ is {\it pendant} if it can be detached from $T$ by deleting a single edge. For $n\geq 2$, let $C = (\ell_1,\ell_2\ldots,\ell_n)$ be a sequence of distinct taxa in $X$. 
We call $C$ an $n$-chain of $T$ if there exists a walk $p_1,p_2,\ldots,p_n$ in $T$ and the elements in $p_2,p_3,\ldots,p_{n-1}$ are all pairwise distinct. Note that $\ell_1$ and $\ell_2$ may have a common parent and $\ell_{n-1}$ and $\ell_n$ may have a common parent. Furthermore, if  $p_1 = p_2$ or $p_{n-1} = p_n$
then $C$ is said to be {\it pendant} in $T$. To ease reading, we sometimes write $C$ to denote the set $\{\ell_1,\ell_2,\ldots,\ell_n\}$. It will always be clear from the context whether $C$ refers to the associated sequence or set of taxa. If a pendant subtree $S$ (resp. an $n$-chain $C$) exists in two phylogenetic trees $T$ and $T'$ on $X$, we say that $S$ (resp. $C$) is a {\it common} subtree (resp. chain) of $T$ and $T'$.

Let $F=\{B_1,B_1,B_2,\ldots,B_k\}$ be an agreement forest for two phylogenetic trees $T$ and $T'$ on $X$, and let $Y$ be a subset of $X$. We say that $Y$ is {\it preserved} in $F$ if there exists an element $B_i$ in $F$ with $i\in\{1, 2,\ldots,k\}$ such that $Y\subseteq B_i$. Later in the article we will make use of  the following theorem from \cite{KelkL19}, referred to as the {\it chain preservation theorem}.

\begin{Theorem}[\cite{KelkL19}]
\label{thm:allchainsintact}
Let $T$ and $T'$ be two phylogenetic trees on $X$. 
Let $K$ be an (arbitrary) set of mutually taxa-disjoint chains that are common to $T$ and $T'$.
Then there exists a maximum agreement forest $F$ of $T$ and $T'$ such that 
\begin{enumerate}
\item every $n$-chain in $K$ with $n\geq 3$
is preserved in $F$, and
\item every 2-chain in $K$ 
that is pendant in at least one of $T$ and $T'$
is preserved in $F$.
\end{enumerate}
\end{Theorem}

\newpage

\section{NP-hardness}\label{Section NP Hard}

In this section we will prove the following theorem.

\begin{Theorem} \label{Thm uMAF is NP-Complete}
    uMAF on caterpillars is NP-complete.
\end{Theorem}

Recall that an \emph{independent set} of an undirected graph $G=(V,E)$ is a set $I \subseteq V$ of mutually non-adjacent vertices. The problem of computing a maximum-size independent set (MIS) is a well-known NP-hard and APX-hard problem. We will establish our theorem by reducing from MIS on cubic graphs. It is well-known that in a cubic graph, the size of a MIS is at least $|V|/4$. Moreover, MIS remains NP-hard and APX-hard on cubic graphs \cite{alimonti2000some}.

Before we can prove our theorem we need to establish a lemma and an observation and describe how two caterpillars $T_G$ and $T'_G$ are built from a cubic graph $G$, which is the input to the MIS problem. 

Let $G=(V,E)$ be a cubic graph, where $n=|V|$ and $m=|E|=3n/2$. For each $v \in V$ we introduce three taxa $v_1, v_2, v_3$. For each vertex $v \in V$ and for each edge $e$ incident to $v$, we introduce two taxa $e^{v\rightarrow},  e^{v\leftarrow}$. Hence, there are in total $3n + (3n \cdot 2) = 9n$ taxa used to encode the actual graph. There will be an additional $6n + 6m= 6n+9n=15n$ taxa introduced that have an auxiliary function, so $24n$ taxa in total. The entire construction is as follows. See also Figure \ref{Fig example T and T' general}.

\begin{Definition} \label{Def reduction trees}
    Let $G=(V,E)$ be a cubic graph. 
    \begin{enumerate}
        \item[$A_v$]: For each vertex $v \in V$, let $e_1, e_2, e_3$ be the three edges it is incident to. We construct a chain $A_v$ as follows: $( e_1^{v\leftarrow}, v_1, e_1^{v\rightarrow}, 
        e_2^{v\leftarrow}, v_2, e_2^{v\rightarrow},
        e_3^{v\leftarrow}, v_3, e_3^{v\rightarrow} )$.

        \item[$B_e$]:    For each edge $e = \{u,v\}$, let $B_e$ be a chain $(e^{u\leftarrow},  e^{v\leftarrow},
        e^{u\rightarrow}, e^{v\rightarrow})$.

        \item[$D_v$]:    For each $v \in V$, let $D_v$ be a chain with the taxa $(v_1, v_2, v_3)$.

        \item[$C_i$]:    For $1\leq i \leq 2(n+m)$, let $C_i$ be a chain with 3 taxa with arbitrary labels. These chains will be mirrored (i.e. have opposite orientations) in 
        $T_G$ and $T'_G$.
        \item[$T_G$] :   Let $T_G$ be a caterpillar alternating each $A_v$ with two $C_i$; this uses the first $2n$ $C_i$ chains. Note that these chains are always in pairs of the form $C_i, C_{i+1}$ ($i$ odd). This is then followed by a block of the remaining $2m = 3n$ $C_i$ chains. 
                  \item[$T'_G$]:    Let $T'_G$ be a caterpillar which alternates each $D_v$ with two $C_i$, where the $C_i$ are mirrored with respect to their orientation in $T_G$. The second part consists of alternating each $B_e$ with two $C_i$, which are again mirrored with respect to their orientation in $T_G$. The $C_i$ chains in $T'_G$ are in the same pairs as in $T_G$.
    \end{enumerate}
\end{Definition}

\begin{figure}[!h]
\begin{center}
\begin{tikzpicture}[scale = 0.8]    
    \node (T1) at (2, 0) {$T_G$};
    \node (Au) [draw, minimum width=0.5cm, minimum height=0.5cm] at (4, 0) {$A_1$};
    \node (C1) [draw, minimum width=0.5cm, minimum height=0.5cm] at (5, 0) {$\overrightarrow{C_1}$};
    \node (C2) [draw, minimum width=0.5cm, minimum height=0.5cm] at (6, 0) {$\overrightarrow{C_2}$};

    \node (Av) [draw, minimum width=0.5cm, minimum height=0.5cm] at (8, 0) {$A_i$};
    \node (C3) [draw, minimum width=0.5cm, minimum height=0.5cm] at (9.5, 0) {$\overrightarrow{C_{2i-1}}$};
    \node (C4) [draw, minimum width=0.5cm, minimum height=0.5cm] at (11, 0) {$\overrightarrow{C_{2i}}$};
    
    \node (Aw) [draw, minimum width=0.5cm, minimum height=0.5cm] at (13, 0) {$A_n$};
    \node (C5) [draw, minimum width=0.5cm, minimum height=0.5cm] at (14.5, 0) {$\overrightarrow{C_{2n-1}}$};
    \node (C6) [draw, minimum width=0.5cm, minimum height=0.5cm] at (16, 0) {$\overrightarrow{C_{2n}}$};
    
    \draw [line width = 1pt]
    (Au) edge (C1)
    (C1) edge (C2)
    (C2) edge [dashed] (Av)
    (Av) edge (C3)
    (C3) edge (C4)
    (C4) edge [dashed] (Aw)
    (Aw) edge (C5)
    (C5) edge (C6)
    
    (C6) edge [dashed] (17.5, 0);

    \node (Au) [draw, minimum width=0.5cm, minimum height=0.5cm] at (5, -1) {$\overrightarrow{C_{2n+1}}$};
    \node (C1) [draw, minimum width=0.5cm, minimum height=0.5cm] at (7.5, -1) {$\overrightarrow{C_{5n}}$};

    \draw [line width = 1pt]
    (3.5,-1) edge [dashed] (Au)
    (Au) edge [dashed] (C1);
    
    \node (T1) at (2,-3) {$T'_G$};
    \node (Au) [draw, minimum width=0.5cm, minimum height=0.5cm] at (4, -2.5) {$D_1$};
    \node (C1) [draw, minimum width=0.5cm, minimum height=0.5cm] at (5, -2.5) {$\overleftarrow{C_1}$};
    \node (C2) [draw, minimum width=0.5cm, minimum height=0.5cm] at (6, -2.5) {$\overleftarrow{C_2}$};

    \node (Av) [draw, minimum width=0.5cm, minimum height=0.5cm] at (8, -2.5) {$D_i$};
    \node (C3) [draw, minimum width=0.5cm, minimum height=0.5cm] at (9.5, -2.5) {$\overleftarrow{C_{2i-1}}$};
    \node (C4) [draw, minimum width=0.5cm, minimum height=0.5cm] at (11, -2.5) {$\overleftarrow{C_{2i}}$};
    
    \node (Aw) [draw, minimum width=0.5cm, minimum height=0.5cm] at (13, -2.5) {$D_n$};
    \node (C5) [draw, minimum width=0.5cm, minimum height=0.5cm] at (14.5, -2.5) {$\overleftarrow{C_{2n-1}}$};
    \node (C6) [draw, minimum width=0.5cm, minimum height=0.5cm] at (16, -2.5) {$\overleftarrow{C_{2n}}$};
    
    \draw [line width = 1pt]
    (Au) edge (C1)
    (C1) edge (C2)
    (C2) edge [dashed] (Av)
    (Av) edge (C3)
    (C3) edge (C4)
    (C4) edge [dashed] (Aw)
    (Aw) edge (C5)
    (C5) edge (C6)
    
    (C6) edge [dashed] (17.5, -2.5);

    \node (Au) [draw, minimum width=0.5cm, minimum height=0.5cm] at (5, -3.5) {$B_1$};
    \node (C1) [draw, minimum width=0.5cm, minimum height=0.5cm] at (6.5, -3.5) {$\overleftarrow{C_{2n+1}}$};
    \node (C2) [draw, minimum width=0.5cm, minimum height=0.5cm] at (8.5, -3.5) {$\overleftarrow{C_{2(n+1)}}$};

    \node (Av) [draw, minimum width=0.5cm, minimum height=0.5cm] at (10.5, -3.5) {$B_j$};
    \node (C3) [draw, minimum width=0.5cm, minimum height=0.5cm] at (12.25, -3.5) {$\overleftarrow{C_{2(n+j)-1}}$};
    \node (C4) [draw, minimum width=0.5cm, minimum height=0.5cm] at (14.5, -3.5) {$\overleftarrow{C_{2(n+j)}}$};
    
    \node (Aw) [draw, minimum width=0.5cm, minimum height=0.5cm] at (16.5, -3.5) {$B_{m}$};
    \node (C5) [draw, minimum width=0.5cm, minimum height=0.5cm] at (18, -3.5) {$\overleftarrow{C_{5n-1}}$};
    \node (C6) [draw, minimum width=0.5cm, minimum height=0.5cm] at (19.5, -3.5) {$\overleftarrow{C_{5n}}$};
    
    \draw [line width = 1pt]
    (3.5,-3.5) edge [dashed] (Au)
    (Au) edge (C1)
    (C1) edge (C2)
    (C2) edge [dashed] (Av)
    (Av) edge (C3)
    (C3) edge (C4)
    (C4) edge [dashed] (Aw)
    (Aw) edge (C5)
    (C5) edge (C6);

\end{tikzpicture}
\caption{The caterpillars $T_G$ and $T'_G$ constructed from $G$. Both caterpillars have $24n$ taxa in total, where $n$ is the number of vertices in $G$.}
\label{Fig example T and T' general}
\end{center}
\end{figure}
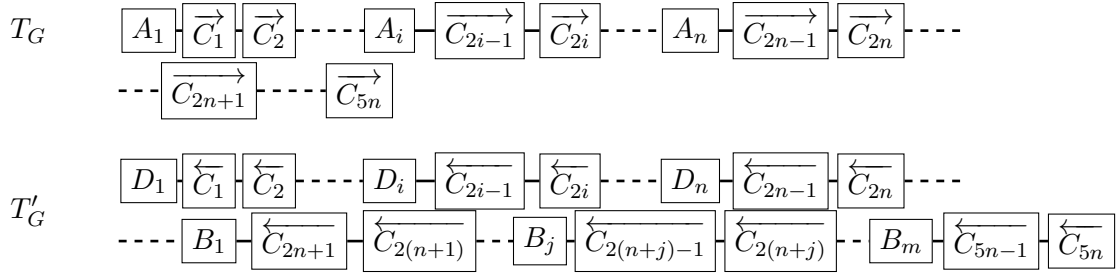

    In the hardness proof we will, given an uMAF $F$ of $T_G$ and $T'_G$, construct an independent set $I_F$ from it. $A_v$ will be used to determine whether vertex $v$ in $G$ is in this independent set or not. For each edge $e=\{u,v\}$, $B_e$ will be used to ensure that $u$ and $v$ are not both in the independent set.
    We return to this point later. First let us look at a small example for $G=K_4$. See Figures \ref{Fig example k4 general} and the corresponding extensions in Figures \ref{Fig example k4 extention T} and \ref{Fig example k4 extention T'}.
    
\begin{figure}[h!]
\begin{center}
\begin{tikzpicture}[scale = 0.8]
    \node (G) at (0,0) {$K_4$};
    \vertex [label=below:$c$] (u) at (-1,-3) {};
    \vertex [label=below:$b$] (v) at (1,-3) {};
    \vertex [label=above:$a$] (w) at (0,-1) {};
    \vertex [label=below:$v$] (x) at (0,-2) {};
    \draw [line width = 1pt]
    (x) edge (v)
    (x) edge (u)
    (x) edge (w)
    (u) edge (v)
    (v) edge (w)
    (w) edge (u);    
    	\draw [thick]
    	
	(1,-1.5) edge [->] node [label=above:Def \ref{Def reduction trees} ] {} (3,-1.5);
    
    \node (T1) at (10, 2) {$T_G$};
    \node (Au) [draw, minimum width=0.5cm, minimum height=0.5cm] at (4, 1) {$A_v$};
    \node (C1) [draw, minimum width=0.5cm, minimum height=0.5cm] at (5, 1) {$\overrightarrow{C_1}$};
    \node (C2) [draw, minimum width=0.5cm, minimum height=0.5cm] at (6, 1) {$\overrightarrow{C_2}$};

    \node (Av) [draw, minimum width=0.5cm, minimum height=0.5cm] at (7.5, 1) {$A_a$};
    \node (C3) [draw, minimum width=0.5cm, minimum height=0.5cm] at (8.5, 1) {$\overrightarrow{C_3}$};
    \node (C4) [draw, minimum width=0.5cm, minimum height=0.5cm] at (9.5, 1) {$\overrightarrow{C_4}$};
    
    \node (Aw) [draw, minimum width=0.5cm, minimum height=0.5cm] at (11, 1) {$A_b$};
    \node (C5) [draw, minimum width=0.5cm, minimum height=0.5cm] at (12, 1) {$\overrightarrow{C_5}$};
    \node (C6) [draw, minimum width=0.5cm, minimum height=0.5cm] at (13, 1) {$\overrightarrow{C_6}$};
    
    \node (Ax) [draw, minimum width=0.5cm, minimum height=0.5cm] at (14.5, 1) {$A_c$};
    \node (C7) [draw, minimum width=0.5cm, minimum height=0.5cm] at (15.5, 1) {$\overrightarrow{C_7}$};
    \node (C8) [draw, minimum width=0.5cm, minimum height=0.5cm] at (16.5, 1) {$\overrightarrow{C_8}$};
    
    \draw [line width = 1pt]
    (Au) edge (C1)
    (C1) edge (C2)
    (C2) edge (Av)
    (Av) edge (C3)
    (C3) edge (C4)
    (C4) edge (Aw)
    (Aw) edge (C5)
    (C5) edge (C6)
    (C6) edge (Ax)
    (Ax) edge (C7)
    (C7) edge (C8)

    (C8) edge [dashed] (17.5, 1);

    \node (Au) [draw, minimum width=0.5cm, minimum height=0.5cm] at (5, 0) {$\overrightarrow{C_{9}}$};
    \node (C1) [draw, minimum width=0.5cm, minimum height=0.5cm] at (6, 0) {$\overrightarrow{C_{10}}$};
    \node (C2) [draw, minimum width=0.5cm, minimum height=0.5cm] at (7, 0) {$\overrightarrow{C_{11}}$};

    \node (Av) [draw, minimum width=0.5cm, minimum height=0.5cm] at (8, 0) {$\overrightarrow{C_{12}}$};
    \node (C3) [draw, minimum width=0.5cm, minimum height=0.5cm] at (9, 0) {$\overrightarrow{C_{13}}$};
    \node (C4) [draw, minimum width=0.5cm, minimum height=0.5cm] at (10, 0) {$\overrightarrow{C_{14}}$};
    
    \node (Aw) [draw, minimum width=0.5cm, minimum height=0.5cm] at (11, 0) {$\overrightarrow{C_{15}}$};
    \node (C5) [draw, minimum width=0.5cm, minimum height=0.5cm] at (12, 0) {$\overrightarrow{C_{16}}$};
    \node (C6) [draw, minimum width=0.5cm, minimum height=0.5cm] at (13, 0) {$\overrightarrow{C_{17}}$};
    
    \node (Ax) [draw, minimum width=0.5cm, minimum height=0.5cm] at (14, 0) {$\overrightarrow{C_{18}}$};
    \node (C7) [draw, minimum width=0.5cm, minimum height=0.5cm] at (15, 0) {$\overrightarrow{C_{19}}$};
    \node (C8) [draw, minimum width=0.5cm, minimum height=0.5cm] at (16, 0) {$\overrightarrow{C_{20}}$};
    
    \draw [line width = 1pt]
    (3.5, 0) edge [dashed] (Au)
    (Au) edge (C1)
    (C1) edge (C2)
    (C2) edge (Av)
    (Av) edge (C3)
    (C3) edge (C4)
    (C4) edge (Aw)
    (Aw) edge (C5)
    (C5) edge (C6)
    (C6) edge (Ax)
    (Ax) edge (C7)
    (C7) edge (C8);
    
    \node (T1) at (10,-1) {$T'_G$};
    \node (Au) [draw, minimum width=0.5cm, minimum height=0.5cm] at (4, -2) {$D_v$};
    \node (C1) [draw, minimum width=0.5cm, minimum height=0.5cm] at (5, -2) {$\overleftarrow{C_1}$};
    \node (C2) [draw, minimum width=0.5cm, minimum height=0.5cm] at (6, -2) {$\overleftarrow{C_2}$};

    \node (Av) [draw, minimum width=0.5cm, minimum height=0.5cm] at (7.5, -2) {$D_a$};
    \node (C3) [draw, minimum width=0.5cm, minimum height=0.5cm] at (8.5, -2) {$\overleftarrow{C_3}$};
    \node (C4) [draw, minimum width=0.5cm, minimum height=0.5cm] at (9.5, -2) {$\overleftarrow{C_4}$};
    
    \node (Aw) [draw, minimum width=0.5cm, minimum height=0.5cm] at (11, -2) {$D_b$};
    \node (C5) [draw, minimum width=0.5cm, minimum height=0.5cm] at (12, -2) {$\overleftarrow{C_5}$};
    \node (C6) [draw, minimum width=0.5cm, minimum height=0.5cm] at (13, -2) {$\overleftarrow{C_6}$};
    
    \node (Ax) [draw, minimum width=0.5cm, minimum height=0.5cm] at (14.5, -2) {$A_c$};
    \node (C7) [draw, minimum width=0.5cm, minimum height=0.5cm] at (15.5, -2) {$\overleftarrow{C_7}$};
    \node (C8) [draw, minimum width=0.5cm, minimum height=0.5cm] at (16.5, -2) {$\overleftarrow{C_8}$};
    
    \draw [line width = 1pt]
    (Au) edge (C1)
    (C1) edge (C2)
    (C2) edge (Av)
    (Av) edge (C3)
    (C3) edge (C4)
    (C4) edge (Aw)
    (Aw) edge (C5)
    (C5) edge (C6)
    (C6) edge (Ax)
    (Ax) edge (C7)
    (C7) edge (C8)
    
    (C8) edge [dashed] (17.5, -2);

    \node (Au) [draw, minimum width=0.5cm, minimum height=0.5cm] at (5, -3) {$B_{\{v,a\}}$};
    \node (C1) [draw, minimum width=0.5cm, minimum height=0.5cm] at (6.5, -3) {$\overleftarrow{C_{9}}$};
    \node (C2) [draw, minimum width=0.5cm, minimum height=0.5cm] at (7.5, -3) {$\overleftarrow{C_{10}}$};

    \node (Av) [draw, minimum width=0.5cm, minimum height=0.5cm] at (9, -3) {$B_{\{v,b\}}$};
    \node (C3) [draw, minimum width=0.5cm, minimum height=0.5cm] at (10.5, -3) {$\overleftarrow{C_{11}}$};
    \node (C4) [draw, minimum width=0.5cm, minimum height=0.5cm] at (11.5, -3) {$\overleftarrow{C_{12}}$};
    
    \node (Aw) [draw, minimum width=0.5cm, minimum height=0.5cm] at (13, -3) {$B_{\{v,c\}}$};
    \node (C5) [draw, minimum width=0.5cm, minimum height=0.5cm] at (14.5, -3) {$\overleftarrow{C_{13}}$};
    \node (C6) [draw, minimum width=0.5cm, minimum height=0.5cm] at (15.5, -3) {$\overleftarrow{C_{14}}$};
    
    \draw [line width = 1pt]
    (3.5, -3) edge [dashed] (Au)
    (Au) edge (C1)
    (C1) edge (C2)
    (C2) edge (Av)
    (Av) edge (C3)
    (C3) edge (C4)
    (C4) edge (Aw)
    (Aw) edge (C5)
    (C5) edge (C6)
    
    (C6) edge [dashed] (17, -3);

    \node (Au) [draw, minimum width=0.5cm, minimum height=0.5cm] at (5, -4) {$B_{\{a,b\}}$};
    \node (C1) [draw, minimum width=0.5cm, minimum height=0.5cm] at (6.5, -4) {$\overleftarrow{C_{15}}$};
    \node (C2) [draw, minimum width=0.5cm, minimum height=0.5cm] at (7.5, -4) {$\overleftarrow{C_{16}}$};

    \node (Av) [draw, minimum width=0.5cm, minimum height=0.5cm] at (9, -4) {$B_{\{a,c\}}$};
    \node (C3) [draw, minimum width=0.5cm, minimum height=0.5cm] at (10.5, -4) {$\overleftarrow{C_{17}}$};
    \node (C4) [draw, minimum width=0.5cm, minimum height=0.5cm] at (11.5, -4) {$\overleftarrow{C_{18}}$};
    
    \node (Aw) [draw, minimum width=0.5cm, minimum height=0.5cm] at (13, -4) {$B_{\{b,c\}}$};
    \node (C5) [draw, minimum width=0.5cm, minimum height=0.5cm] at (14.5, -4) {$\overleftarrow{C_{19}}$};
    \node (C6) [draw, minimum width=0.5cm, minimum height=0.5cm] at (15.5, -4) {$\overleftarrow{C_{20}}$};
    
    \draw [line width = 1pt]
    (3.5, -4) edge [dashed] (Au)
    (Au) edge (C1)
    (C1) edge (C2)
    (C2) edge (Av)
    (Av) edge (C3)
    (C3) edge (C4)
    (C4) edge (Aw)
    (Aw) edge (C5)
    (C5) edge (C6);

\end{tikzpicture}
\caption{The construction of $T_G$ and $T'_G$ when $G=K_4$.}
\label{Fig example k4 general}
\end{center}
\end{figure}
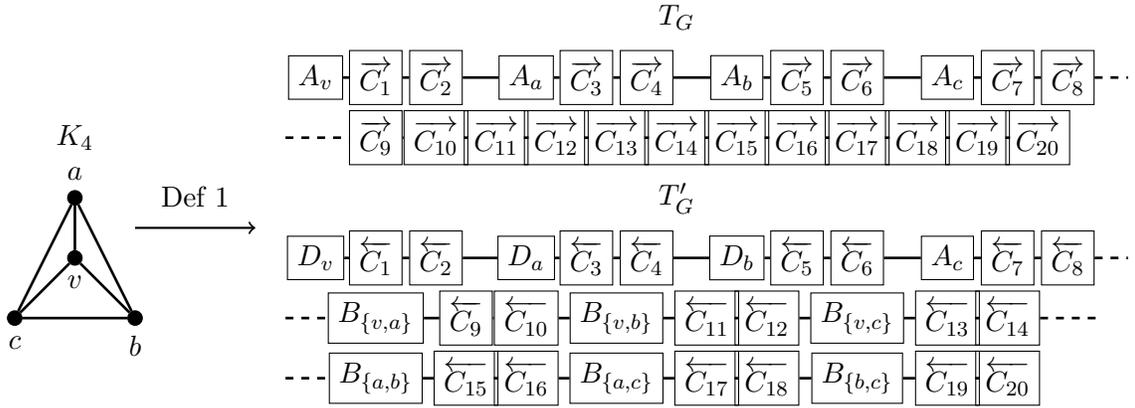

The following lemma will show that an agreement forest
$F$ on $T_G$ and $T_G'$ that preserves all chains $C_i$ has nice properties which we can use to build an independent set $I_F$.

\begin{Lemma} \label{Lemma Av unique 6 pieces}
Let $F$ be an arbitrary agreement forest that preserves all the $C_i$ chains. (i) If a component $B \in F$ intersects with some $A_v$, then $B \subseteq A_v$. Similarly: if $B$ intersects with some $B_e$, then $B \subseteq B_e$, and if $B$ intersects with some $D_v$, then $B \subseteq D_v$; (ii) for each vertex $v \in V$, at least six components of $F$ are required to cover all taxa in $A_v$, and there is only one way to cover $A_v$ with 6 components.
\end{Lemma}

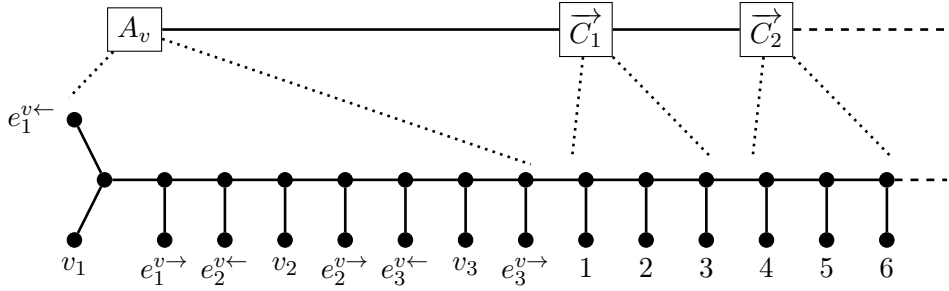
\begin{figure}[!h]
\begin{center}
\begin{tikzpicture}[scale = 0.8]
    
    \node (Au) [draw, minimum width=0.5cm, minimum height=0.5cm] at (9.5, 2.5) {$A_v$};
    \node (C1) [draw, minimum width=0.5cm, minimum height=0.5cm] at (17, 2.5) {$\overrightarrow{C_1}$};
    \node (C2) [draw, minimum width=0.5cm, minimum height=0.5cm] at (20, 2.5) {$\overrightarrow{C_2}$};
    
    \draw [line width = 1pt]
    (Au) edge (C1)
    (C1) edge (C2)
    (C2) edge [dashed] (23, 2.5);
    
    \vertex [label=below:$$] (v1p) at (9,0) {};
    \vertex [label=below:$$] (v2p) at (10,0) {};    
    \vertex [label=below:$$] (w1p) at (11,0) {};    
    \vertex [label=below:$$] (wup) at (12,0) {};    
    \vertex [label=below:$$] (w2p) at (13,0) {};    
    \vertex [label=below:$$] (x1p) at (14,0) {};    
    \vertex [label=below:$$] (xup) at (15,0) {};    
    \vertex [label=below:$$] (x2p) at (16,0) {};
    
    \vertex [label=left:$e_1^{v\leftarrow}$] (v1) at (8.5,1) {};
    \vertex [label=below:$v_1$] (vu) at (8.5,-1) {};
    \vertex [label=below:$e_1^{v\rightarrow}$] (v2) at (10,-1) {};
    
    \vertex [label=below:$e_2^{v\leftarrow}$] (w1) at (11,-1) {};
    \vertex [label=below:$v_2$] (wu) at (12,-1) {};
    \vertex [label=below:$e_2^{v\rightarrow}$] (w2) at (13,-1) {};
    
    \vertex [label=below:$e_3^{v\leftarrow}$] (x1) at (14,-1) {};
    \vertex [label=below:$v_3$] (xu) at (15,-1) {};
    \vertex [label=below:$e_3^{v\rightarrow}$] (x2) at (16,-1) {};
    
    \node (x2pd) at (16.25,0.2) {};
    \node (v1d) at (8.25,1.2) {};

    \draw [line width = 1pt]
    (Au) edge [dotted] (v1d)
    (Au) edge [dotted] (x2pd)
    
    (v1p) edge (v1)
    (v1p) edge (vu)
    (v2p) edge (v2)
    
    (w1p) edge (w1)
    (wup) edge (wu)
    (w2p) edge (w2)
    
    (x1p) edge (x1)
    (xup) edge (xu)
    (x2p) edge (x2);

    \vertex [label=below:$$] (c1p) at (17,0) {};
    \vertex [label=below:$$] (c2p) at (18,0) {};    
    \vertex [label=below:$$] (c3p) at (19,0) {};    
    \vertex [label=below:$$] (c4p) at (20,0) {};    
    \vertex [label=below:$$] (c5p) at (21,0) {};    
    \vertex [label=below:$$] (c6p) at (22,0) {};
    
    \vertex [label=below:$1$] (c1) at (17,-1) {};
    \vertex [label=below:$2$] (c2) at (18,-1) {};
    \vertex [label=below:$3$] (c3) at (19,-1) {};
    
    \vertex [label=below:$4$] (c4) at (20,-1) {};
    \vertex [label=below:$5$] (c5) at (21,-1) {};
    \vertex [label=below:$6$] (c6) at (22,-1) {};
    
    \node (c1pd) at (16.75,0.2) {};
    \node (c3pd) at (19.25,0.2) {};
    
    \node (c4pd) at (19.75,0.2) {};
    \node (c6pd) at (22.25,0.2) {};
    
    \draw [line width = 1pt]
    (C1) edge [dotted] (c1pd)
    (C1) edge [dotted] (c3pd)

    (C2) edge [dotted] (c4pd)
    (C2) edge [dotted] (c6pd)
    
    (v1p) edge (c6p)
    (c1p) edge (c1)
    (c2p) edge (c2)
    (c3p) edge (c3)
    (c4p) edge (c4)
    (c5p) edge (c5)
    (c6p) edge (c6)
    (c6p) edge [dashed] (23,0);
	
\end{tikzpicture}
\caption{Extension of $A_v$, $C_1$ and $C_2$ from Figure \ref{Fig example k4 general}. }
\label{Fig example k4 extention T}
\end{center}
\end{figure}

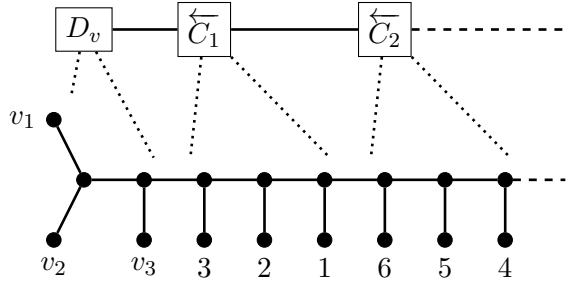
\begin{figure}[!h]
\begin{center}
\begin{tikzpicture}[scale = 0.8]
    
    \node (Au) [draw, minimum width=0.5cm, minimum height=0.5cm] at (9, 5.5) {$D_v$};
    \node (C1) [draw, minimum width=0.5cm, minimum height=0.5cm] at (11, 5.5) {$\overleftarrow{C_1}$};
    \node (C2) [draw, minimum width=0.5cm, minimum height=0.5cm] at (14, 5.5) {$\overleftarrow{C_2}$};
    
    \draw [line width = 1pt]
    (Au) edge (C1)
    (C1) edge (C2)
    (C2) edge [dashed] (17, 5.5);
    
    \vertex [label=below:$$] (v1p) at (9,3) {};
    \vertex [label=below:$$] (v2p) at (10,3) {};
    
    \vertex [label=left:$v_1$] (v1) at (8.5,4) {};
    \vertex [label=below:$v_2$] (vu) at (8.5,2) {};
    \vertex [label=below:$v_3$] (v2) at (10,2) {};

    \node (v1d) at (8.75,4.2) {};
    \node (v2pd) at (10.25,3.2) {};

    \draw [line width = 1pt]
    (Au) edge [dotted] (v1d)
    (Au) edge [dotted] (v2pd)
    
    (v1p) edge (v1)
    (v1p) edge (vu)
    (v2p) edge (v2);

    \vertex [label=below:$$] (c1p) at (11,3) {};
    \vertex [label=below:$$] (c2p) at (12,3) {};    
    \vertex [label=below:$$] (c3p) at (13,3) {};    
    \vertex [label=below:$$] (c4p) at (14,3) {};    
    \vertex [label=below:$$] (c5p) at (15,3) {};    
    \vertex [label=below:$$] (c6p) at (16,3) {};
    
    \vertex [label=below:$3$] (c1) at (11,2) {};
    \vertex [label=below:$2$] (c2) at (12,2) {};
    \vertex [label=below:$1$] (c3) at (13,2) {};
    
    \vertex [label=below:$6$] (c4) at (14,2) {};
    \vertex [label=below:$5$] (c5) at (15,2) {};
    \vertex [label=below:$4$] (c6) at (16,2) {};
    
    \node (c1pd) at (10.75,3.2) {};
    \node (c3pd) at (13.25,3.2) {};
    
    \node (c4pd) at (13.75,3.2) {};
    \node (c6pd) at (16.25,3.2) {};
    
    \draw [line width = 1pt]
    (C1) edge [dotted] (c1pd)
    (C1) edge [dotted] (c3pd)

    (C2) edge [dotted] (c4pd)
    (C2) edge [dotted] (c6pd)
    
    (v1p) edge (c6p)
    (c1p) edge (c1)
    (c2p) edge (c2)
    (c3p) edge (c3)
    (c4p) edge (c4)
    (c5p) edge (c5)
    (c6p) edge (c6)
    (c6p) edge [dashed] (17,3);
	
\end{tikzpicture}
\caption{Extension of $D_v$, $C_1$, $C_2$ from Figure \ref{Fig example k4 general}. }
\label{Fig example k4 extention T'}
\end{center}
\end{figure}

\begin{proof}
Observe that, as all the $C_i$ chains are preserved, every $C_i$ appears as a component of $F$. To see why, note that chain $C_i$ ($i$ odd) cannot be in a component together with the chain $C_{i+1}$. This is because of the opposite orientations of these two chains in $T_G$ and $T'_G$. A component
$C_i$ ($i$ odd) cannot be in a component with taxa that are to its left in $T_G$, because these taxa would have to be on its right in $T'_G$ - but this would violate the preservation of $C_{i+1}$. Symmetrically, chain $C_{i+1}$ cannot be together in a component with taxa that are to its right in $T_G$. Hence, all the $C_i$ are components (without any additional taxa) in $F$. Combining this with the fact each $A_v$ chain is sandwiched between two $C_i$ chains in $T_G$, proves that a component cannot simultaneously intersect with $A_v$ and something outside $A_v$. Exactly the same reasoning holds for the $D_v$ and $B_e$ sets, since these are also sandwiched between pairs of $C_i$ chains. This establishes (i). Towards (ii), let $v$ be an arbitrary vertex in $V$ and let $e_1, e_2, e_3$ be the three edges it is incident to. Observe that $A_v$ consists of the \emph{disjoint} union
of the following four sets: $(A_v \cap D_v), (A_v \cap B_{e_1}), (A_v \cap B_{e_2}), (A_v \cap B_{e_3})$, containing 3,2,2,2 taxa respectively. Due to (i), any component intersecting with one of these four sets, cannot intersect any of the other four sets. Now, if $B$ contains all three taxa from $D_v$, then two components each are required to cover the three $(A_v \cap B_{e_i})$ sets; so 7 components in total. It can be easily verified in a similar way that if $B$ contains 2 taxa from $D_v$, 7 components are required, and that if $B$ contains exactly 1 taxon from $D_v$ then at least 6 components are required - and that in fact the only way to cover $A_v$ with 6 components is as follows:

\[
\{v_1\}, \{v_2\}, \{v_3\},  \{ e_1^{v\leftarrow}, e_1^{v\rightarrow} \}, 
\{ e_2^{v\leftarrow}, e_2^{v\rightarrow} \}, 
\{ e_3^{v\leftarrow}, e_3^{v\rightarrow} \}.
\]
    
\end{proof}

The following observation, which we state without proof, will be useful for both the NP-hardness and APX-hardness reductions.

\begin{Observation}
\label{obs:replace}
Let $F$ be an arbitrary agreement forest that preserves all the $C_i$ chains. If an $A_v$ set is
covered by 7 or more components, then deleting all the components that intersect with $A_v$
and adding the following components yields a new agreement forest $F'$ such that $|F'| \leq |F|$ and
all the $C_i$ chains are still preserved in $F'$:
\[
\{v_1, v_2,v_3\},  \{ e_1^{v\leftarrow} \}, \{ e_1^{v\rightarrow} \}, 
\{ e_2^{v\leftarrow} \}, \{ e_2^{v\rightarrow} \}, 
\{ e_3^{v\leftarrow} \}, \{e_3^{v\rightarrow} \}.
\]

\end{Observation}

\noindent
We can now move on to the proof of Theorem \ref{Thm uMAF is NP-Complete}. The high-level idea is that if an $A_v$ set is covered by 6 components, then (by Lemma \ref{Lemma Av unique 6 pieces}) there is a unique way of doing this, which in turn 
enforces that each $A_u$ set of a neighbour $u$ of $v$ requires at least 7 components.
The neighbouring $A_u$ can then be assumed via Observation \ref{obs:replace} to consist of the 7 components named in that
observation. In this way the selection of independent sets with many elements is preferred.
\begin{proof}

    For the input graph $G$ to the MIS problem we construct $T_G$ and $T'_G$ as described in Definition \ref{Def reduction trees}. We let $k$ be the size of a maximum independent set of $G$
    and write $Opt(T_G, T'_G)$ to denote the size of a uMAF of $T_G$ and $T'_G$. We prove that
    $Opt(T_G, T'_G) = 12n-k$.
    
    We start by showing that $Opt(T_G, T'_G) \leq 12n-k$. Let $I \subseteq V$ be a MIS of $G$, so $|I|=k$. For each vertex $v \in I$, which has incident edges $e_1, e_2, e_3$, we introduce the 6 components
    \[
\{v_1\}, \{v_2\}, \{v_3\},  \{ e_1^{v\leftarrow}, e_1^{v\rightarrow} \}, 
\{ e_2^{v\leftarrow}, e_2^{v\rightarrow} \}, 
\{ e_3^{v\leftarrow}, e_3^{v\rightarrow} \},
\]
and for each vertex $v \not \in I$ we introduce the 7 components
\[
\{v_1, v_2,v_3\},  \{ e_1^{v\leftarrow} \}, \{ e_1^{v\rightarrow} \}, 
\{ e_2^{v\leftarrow} \}, \{ e_2^{v\rightarrow} \}, 
\{ e_3^{v\leftarrow} \}, \{e_3^{v\rightarrow} \}.
\]
Finally, we add each $C_i$ as its own component. Summarizing, our agreement forest consists of the following components:
    
    \begin{tabular}{l c r}
         number of components corresponding to $v\in I$ && $6k$ \\
         number of components corresponding to $v\notin I$ && $7(n-k)$\\
         number of components corresponding to the $C_i$ chains && $2(n+m)$
    \end{tabular}\\
\\
Recalling that $m=3n/2$ we obtain a total of $6k + 7(n-k) + 2n + 3n = 12n - k$ components. This
concludes the proof that  $Opt(T_G,T'_G) \leq 12n-k$.
    
We now prove $Opt(T_G,T'_G) \geq 12n-k$. 
By applying Theorem \ref{thm:allchainsintact} (the chain preservation theorem) to the $C_i$ chains, we know that there is a maximum agreement forest $F$ of $T_G$  and $T'_G$ in which all the $C_i$ are preserved. Fix such an $F$. As noted in Lemma \ref{Lemma Av unique 6 pieces} the components of $F$ that intersect with a given $A_v$, are all completely contained inside $A_v$. 
Now, we apply Observation \ref{obs:replace} to all $A_v$ that are covered by 7 or more components; this does not increase the size of the forest. The remaining $A_v$ are covered by exactly 6 components each, and from Lemma \ref{Lemma Av unique 6 pieces} there is a unique way of doing this. Observe that it is not possible, whenever $u$ and $v$ are adjacent vertices, and $e$ is the edge between them, to cover both $A_u$ and $A_v$ with 6 components. 
This would require the agreement forest to contain both the components $\{e^{u\leftarrow}, e^{u\rightarrow}\}$ and $\{e^{v\leftarrow}, e^{v\rightarrow}\}$. However, this would mean that $B_e$ is covered by a single component. But $B_e$ contains taxa that are also in $A_u$. Thus, the component that covers $B_e$ must be fully contained in $A_v$ whilst containing taxa not present in $A_u$. This is not possible.
Hence, the $A_v$ that are covered by exactly 6 components point out an independent set.  Let $p$ (respectively, $q$) be the number of $A_v$ in $F$ that are covered by 6 (respectively, 7) components. Giving us $Opt(T_G, T'_G)=  6p + 7q + 2(m+n)$, so
$6p = Opt(T_G, T'_G) - 2(m+n) - 7q$. Now, given that $q = n-p$ we get $6p = Opt(T_G, T'_G) - 3n - 2n - 7n + 7p$ and thus $p = 12n - Opt(T_G, T'_G)$.  Given that $k \geq p$ we are done.

Finally, we note that the uMAF problem is definitely in NP. Given a partition $F$ of the taxa $X$ of the two input trees, it is straightforward to check in time $O(|X|^4)$ that it induces a valid agreement forest. Specifically, we check that the induced subtrees are mutually disjoint in each input tree, and that the set of quartets (i.e. phylogenetic trees on subsets of 4 taxa) induced by each block of the partition is the same in both input trees: two phylogenetic trees are topologically equal if and only if they each induce the same set of induced quartets \cite{SempleS03}. Combined with the above NP-hardness result, this proves that the problem is NP-complete.
\end{proof}

\section{APX-completeness}\label{Section APX Hard}

In the previous section we have proven that uMAF on caterpillars is NP-complete. Now we will prove that uMAF on caterpillars is APX-complete.  The (general) problem is known to have a polynomial-time 3-approximation \cite{WhiddenBZ13,whidden2009unifying} which immediately places the problem in APX. The APX-hardness of the problem on general trees was stated but not proven in \cite{HeinJWZ96}. Our result thus confirms and strengthens this claim to caterpillars. The reduction is essentially the same as NP-hardness proof, with some slight modifications
due to us working with approximation algorithms rather than exact algorithms.

\newpage

Before starting, it is helpful to establish the following corollary to Theorem 
    \ref{thm:allchainsintact}. It differs from that theorem in the sense that it considers agreement
    forests that are not necessarily optimal, and it is explicitly algorithmic.

\begin{Corollary}[\cite{KelkL19}]
\label{cor:allchainsintact}
Let $T$ and $T'$ be two phylogenetic trees on $X$. 
Let $K$ be an (arbitrary) set of mutually taxa-disjoint chains that are common to $T$ and $T'$.
Let $F$ be an arbitrary agreement forest of $T$ and $T'$. There exists an agreement forest $F'$ such
that $|F'| \leq |F|$ and,
\begin{enumerate}
\item every $n$-chain in $K$ with $n\geq 3$
is preserved in $F'$, and
\item every 2-chain in $K$ 
that is pendant in at least one of $T$ and $T'$
is preserved in $F'$.
\end{enumerate}
Also, $F'$ can be constructed from $F$ in polynomial time.
\end{Corollary}
\begin{proof}
Although not stated as such, the proof is implicit in the proof of Theorem \ref{thm:allchainsintact}
given in \cite{KelkL19}. The proof there argues that, if one of the chains in $C \in K$ is not preserved,
then the agreement forest can be explicitly modified such that $C$ is preserved, the total number
of components does not increase, and any chains that were preserved prior to this transformation
are also preserved afterwards. It is easy to check in polynomial time whether a chain in $K$ is
preserved, and the constructive modifications described in the proof can also be easily undertaken in polynomial
time. The only subtlety in the re-use of the proof is that in several places a contradiction is triggered on the assumption that $F$ was a maximum agreement forest. However, a careful reading shows that it is not necessary to use proof by contradiction here at all, and that the assumption that $F$ is maximum
is not required; it was simply an easy way to end the proof. Instead of contradiction, the chain $C$ can be preserved, such that no other preserved chains are damaged, and such that the number of agreements in the agreement forest decreases (which is fine for our purposes).
\end{proof}

\begin{Theorem} \label{Thm uMAF on caterpillars is APX-Complete}
    uMAF on caterpillars is APX-complete.
\end{Theorem}

Corollary \ref{cor:allchainsintact} makes it fairly easy to prove APX-hardness. In particular,
given a (not necessarily optimal) agreement forest $F$ for $T_G$ and $T'_G$, Corollary \ref{cor:allchainsintact} shows that we can construct in polynomial time an agreement forest $F'$ that is no larger than $F$ and in which all the $C_i$ chains are preserved\footnote{We note that we do not actually need the full power of Corollary \ref{cor:allchainsintact} here, since we are only considering chains of length 3, and the chains are arranged in a highly specific way in $T_G, T'_G$. Indeed, it would be an option to give a direct, self-contained proof that all the $C_i$ can be assumed to be preserved. However, this would lead to a tedious case analysis, and would in many respects simply be a repetition of the proof in \cite{KelkL19}.}. 

We  give the definition of an L-reduction \cite{papadimitriou1991optimization} and then show that there exists an L-reduction from MIS on cubic graphs to uMAF on caterpillars, proving that uMAF on caterpillars is APX-hard.

\begin{Definition}
    Define two mappings $f$ and $g$ and two positive constants $\alpha$ and $\beta$ such that the following holds:
    \begin{enumerate}
        \item $f$ is a function that in polynomial time maps the input $G$ to MIS to two trees $T_G$ and $T'_G$ that are the input for uMAF;
        \item For any input $G$ we have: 
            \[
                Opt(T_G, T'_G) \leq \alpha Opt(G)
            \]
	where here $Opt(G)$ denotes the size of a maximum independent set of $G$;
        \item $g$ is a function that maps in polynomial time an agreement forest $F$ for $T_G$ and $T'_G$ to an independent set $I_F$ of $G$;
        \item For any agreement forest $F$ for $T_G$ and $T'_G$ we have: 
            \[
               |Opt(G) - |I_F|| \leq \beta |Opt(T_G, T'_G) - |F||. 
            \]
    \end{enumerate}
    Together this forms an L-reduction from MIS to uMAF.

\end{Definition}

Now for the proof of Theorem \ref{Thm uMAF on caterpillars is APX-Complete}.

\begin{proof}
Let $f$ be the mapping described in Definition \ref{Def reduction trees}. Clearly $T_G$ and $T'_G$ can be constructed in polynomial time. This establishes the first property of the L-reduction.\
    
For the second property, $G$ is a cubic graph so $Opt(G) \geq \frac{1}{4} n$. Let $Opt(G) = k$. Then, as in the proof of theorem \ref{Thm uMAF is NP-Complete} we have $Opt(T_G, T'_G) = 12n-k$. We require $12n - k \leq \alpha k$ which is equivalent to $12\frac{n}{k} \leq \alpha + 1$.  Given that $k \geq n/4$, the left hand side of the  inequality is at most 48. Hence, it is sufficient to select $\alpha = 47$.   
    
For the third property, let $g$ be the polynomial-time mapping defined as follows. Let $F$ be an arbitrary agreement forest of $T_G, T'_G$; this is the input to $g$. We apply Corollary \ref{cor:allchainsintact} to it (letting $K$ be the set of the $C_i$ chains); this yields in polynomial time a new agreement forest $F'$ such that $|F'| \leq |F|$ and in which all the $C_i$ chains are preserved. We then apply Observation \ref{obs:replace} to all $A_v$ in $F'$ which have 7 or more components. Let $F''$ denote this transformed agreement forest; we have $|F''| \leq |F'| \leq |F|$. We create an independent set $I_F \subseteq V$ as follows: $v \in I_F$ if and only if $A_v$ is covered by 6 components in $F''$. As argued in the NP-hardness proof, this will create an independent set.

\newpage

For the fourth property we need to find a value for $\beta$ such that $|Opt(G) - |I_F|| \leq \beta |Opt(T_G, T'_G) - |F||$. Let $\ell = |I_F|$ be the number of $A_v$ in $F''$ that are covered by 6 components. We know that $|F| \geq |F'| \geq |F''|= 12n - \ell$. Observe:
        
    \begin{align*}
        |Opt(T_G, T'_G) - |F|| & = |F| - Opt(T_G, T'_G) \\
                            & \geq |F''| - Opt(T_G, T'_G) \\
                            & = 12n-\ell - Opt(T_G, T'_G) \\
                            & = 12n-\ell - (12n-k) \\
                            & = k-\ell \\
                            & = |Opt(G) - |I_F||
    \end{align*}
    
    So we pick $\beta = 1$ and we are done.\
    
\end{proof}

\section{A tight 7k kernel}\label{Section Kernel}

Recall the definitions of common subtrees and common chains from the preliminaries. It is well-known that the following two polynomial-time reduction rules do not alter the size of the uMAF \cite{AllenS01}:\\

\noindent {\bf Subtree reduction.} If $T$ and $T'$ have a maximal common pendant subtree $S$ with at least two leaves, then reduce $T$ and $T'$ to $T_r$ and $T'_r$, respectively, by replacing $S$ with a single leaf with a new label.\\

\noindent {\bf Chain reduction.} If $T$ and $T'$ have a maximal common $n$-chain $C=(\ell_1,\ell_2,\ldots,\ell_n)$ with $n\geq 4$, then reduce $T$ and $T'$ to $T_r=T|X\setminus \{\ell_4,\ell_5,\ldots,\ell_n\}$ and $T_r'=T'|X\setminus \{\ell_4,\ell_5,\ldots,\ell_n\}$, respectively.\\

When applied to exhaustion on two unrooted binary trees, at which point we say the trees are \emph{fully reduced}, these rules yield an instance with (ignoring additive terms) at most $15k$ taxa \cite{KelkL18}, where $k$ is the size of the uMAF\footnote{The kernel bound given in \cite{KelkL18} is in terms of TBR distance, rather than uMAF, but as noted earlier these quantities only differ by 1, so only additive terms are affected.}, and the analysis is tight.

Note that applying the subtree or chain reduction to a caterpillar produces a new caterpillar. In this section we will show that, when applied to exhaustion on two caterpillars, a much smaller kernel is obtained than on general unrooted binary trees.

\begin{Theorem} \label{Thm uMAF has a 7k kernel}
    There is a 7k kernel for uMAF on caterpillars using only the common chain and subtree reductions, and this is tight up to a constant additive term.
\end{Theorem}

\begin{proof}
 Let $F$ be an uMAF for caterpillars $T$ and $T'$ with $k$ components, where $T$ and $T'$ are fully reduced. We prove that $n \leq 7k$. Suppose that $B$ is a component of $F$ with at least 4 taxa. Observe that in at least one of $T$ and $T'$, say $T$, there is some taxon $x \not \in B$ such that $\{x\} \in F$ (i.e. $x$ is a singleton component in $F$),  $x$ is adjacent to $T[B]$ (i.e. there is an edge $\{x,u\}$ in $T$ such that $u$ is a node of $T[B]$), and every path in $T$ from $x$ to a taxon of $B$ contains at least three edges. If this was not so, then $B$ would be a common chain of length at least 4 and this would contradict the assumption that the chain reduction had been applied to exhaustion. Thus the existence of $B$ forces $x$ to be a singleton component in $F$. We say that $x$ has been \emph{orphaned by} $B$ \emph{in} $T$. 
 Observe:
 \begin{enumerate}
\item  Within a given tree, say $T$, a taxon $x$ can be orphaned by at most one component.
\item A taxon $x$ can be orphaned by in total at most two different components of $F$: one in $T$, and one in $T'$.
\item Within a given tree, say $T$, a component $B \in F$ with at least 4 leaves orphans at least $\lfloor \frac{|B|-1}{3}\rfloor$ taxa.
\end{enumerate}
The third observation is the result of the pigeon-hole principle and the fact, as noted above, that to avoid triggering the chain reduction every sequence of four taxa within a component of $F$ must orphan at least one taxon. So a sequence of 5 taxa needs to orphan at least one taxon, and the same is true for a sequence of 6; a sequence of 7 needs to orphan at least 2 taxa, and so on.

Now, we are ready to bound $n$. Let $x_i$ (where $i \geq 1$) be the number of components in $F$ that contain exactly $i$ taxa. The total number of taxa is thus $\sum_i i \cdot x_i$. To obtain an upper bound, it is sufficient to maximize this sum subject to all $x_i$ being non-negative integers and two constraints:
\[
\sum_i x_i \leq k,
\]
and
\[
\frac{1}{2} \sum_i x_i\bigg \lfloor \frac{|i|-1}{3} \bigg \rfloor \leq x_1.
\]

The second constraint is the result of combining the second and third observations above. (Note that in the summations we could, if desired, take $6k + 3$ as a trivial, finite upper bound on values of $i$ that need to be considered. This is because if $i \geq 6k+4$ then any feasible integral solution to the above constraints must have $x_i = 0$, because of the first constraint.)

As we are only seeking an upper bound on the number of taxa, we can relax the integrality constraints on the $x_i$ variables and to allow them to be fractional. This gives us a linear program (LP). We can use weak duality to place an upper bound on this LP, which is thus an upper bound on the original integral program and thus an upper bound on the size of the kernel (see e.g. \cite{schrijver1998theory} for more background on LP duality). Specifically, we obtain a dual LP\footnote{When building the dual it is helpful to rearrange the second constraint as: $-2x_1 + \sum_{i \geq 2}  x_i \lfloor \frac{|i|-1}{3} \rfloor \leq 0$.} with two dual variables $y_1, y_2$ (corresponding to the two constraints above), an objective function $k \cdot y_1$ and two types of constraints. The first constraint (corresponding to $x_1$ in the original LP) is $y_1 - 2y_2 \geq 1$. For $i \geq 2$, the corresponding dual constraint is $y_1 + \lfloor \frac{|i|-1}{3} \rfloor y_2 \geq i$. It can be verified that taking $y_1 = 7$ and $y_2 = 3$ yields a feasible solution to all dual constraints, achieving an objective function value of $7k$. This completes the proof that the kernel has at most $7k$ taxa.

Now we will prove that this bound is tight up to additive terms by giving for each $k \geq 3$ two fully reduced caterpillars on $7k-8$ leaves, where $k$ is the size of the uMAF. See Figure \ref{fig:tight}. The $\overrightarrow{A}$ chain (respectively, $\overrightarrow{B}$ chain) is the concatenation of all the $\overrightarrow{A_i}$ (respectively, $\overrightarrow{B_i}$) chains. There are $2 \times 3 \times (k-2) + 2*3 = 6k-6$ taxa involved in the $\overrightarrow{A_i}$ and $\overrightarrow{B_i}$ chains, plus $k-2$ $c_i$ taxa, giving $7k-8$ taxa in total. It is easy to verify that $\{ \overrightarrow{A}, \overrightarrow{B}, c_1, \ldots, c_{k-2}\}$ is an agreement forest containing in total $k$ components. It remains only to show that this is an uMAF i.e. that a smaller agreement forest does not exist. Observe, first, that the $\overrightarrow{A_i}$ and $\overrightarrow{B_i}$ form a set of $2(k-1)$ common chains. If we apply Theorem \ref{thm:allchainsintact} to these chains, we see that there exists an uMAF in which these chains are all preserved. Fix such an uMAF. Observe that an $\overrightarrow{A_i}$ chain cannot be in the same component of the uMAF as an $\overrightarrow{B_i}$ chain, because they have conflicting orders in $T$ and $T'$. Next, observe that an $\overrightarrow{A_i}$ or $\overrightarrow{B_i}$ chain cannot be in a component of the uMAF with a $c_i$ taxon: the $c_i$ would be attached to opposite ends of the chain in the two trees. Combined with the fact that all the  $\overrightarrow{A_i}$ and $\overrightarrow{B_i}$
chains are preserved, we conclude that all the $c_i$ taxa form singleton components in the uMAF.  The only remaining freedom in the uMAF is to merge all the $\overrightarrow{A_i}$ into a single component, and all the $\overrightarrow{B_i}$ into a single component. Hence, the uMAF indeed has size $k$.
\end{proof}

\begin{figure}[!h]
\begin{center}
\begin{tikzpicture}[scale = 0.8]    
    
    \node (T1) at (2, 0) {$T$};
    \node (A1) [draw, minimum width=0.5cm, minimum height=0.5cm] at (4, 0) {$\overrightarrow{A_1}$};
    \vertex [label=below:$$] (u1) at (5,0) {};
    \vertex [label=below:$c_1$] (c1) at (5,-1) {};
    \node (A2) [draw, minimum width=0.5cm, minimum height=0.5cm] at (6, 0) {$\overrightarrow{A_2}$};
    \vertex [label=below:$$] (u2) at (7,0) {};
    \vertex [label=below:$c_2$] (c2) at (7,-1) {};

    \node (Ak-1) [draw, minimum width=0.5cm, minimum height=0.5cm] at (10, 0) {$\overrightarrow{A_{k-2}}$};
    \vertex [label=below:$$] (uk) at (11,0) {};
    \vertex [label=below:$c_{k-2}$] (ck) at (11,-1) {};

    \node (Ak) [draw, minimum width=0.5cm, minimum height=0.5cm] at (12, 0) {$\overrightarrow{A_{k-1}}$};
    
    \node (B) [draw, minimum width=0.5cm, minimum height=0.5cm] at (14, 0) {$\overrightarrow{B}$};
    
    \draw [line width = 1pt]
    (A1) edge (u1)
    (u1) edge (c1)
    (u1) edge (A2)
    (A2) edge (u2)
    (u2) edge (c2)
    (u2) edge [dashed] (Ak-1)
    (Ak-1) edge (uk)
    (uk) edge (ck)
    (uk) edge (Ak)
    (Ak) edge (B);

    \node (T1) at (2, -3) {$T'$};
    \node (A1) [draw, minimum width=0.5cm, minimum height=0.5cm] at (4, -3) {$\overrightarrow{B_1}$};
    \vertex [label=below:$$] (u1) at (5,-3) {};
    \vertex [label=below:$c_1$] (c1) at (5,-4) {};
    \node (A2) [draw, minimum width=0.5cm, minimum height=0.5cm] at (6, -3) {$\overrightarrow{B_2}$};
    \vertex [label=below:$$] (u2) at (7,-3) {};
    \vertex [label=below:$c_2$] (c2) at (7,-4) {};

    \node (Ak-1) [draw, minimum width=0.5cm, minimum height=0.5cm] at (10, -3) {$\overrightarrow{B_{k-2}}$};
    \vertex [label=below:$$] (uk) at (11,-3) {};
    \vertex [label=below:$c_{k-2}$] (ck) at (11,-4) {};

    \node (Ak) [draw, minimum width=0.5cm, minimum height=0.5cm] at (12, -3) {$\overrightarrow{B_{k-1}}$};
    
    \node (B) [draw, minimum width=0.5cm, minimum height=0.5cm] at (14, -3) {$\overrightarrow{A}$};
    
    \draw [line width = 1pt]
    (A1) edge (u1)
    (u1) edge (c1)
    (u1) edge (A2)
    (A2) edge (u2)
    (u2) edge (c2)
    (u2) edge [dashed] (Ak-1)
    (Ak-1) edge (uk)
    (uk) edge (ck)
    (uk) edge (Ak)
    (Ak) edge (B);
	
\end{tikzpicture}
\caption{$T$ and $T'$ are fully reduced, have $7k-8$ leaves and an uMAF $F = \{A, B, c_1, ..., c_{k-2}\}$ of size $k$. $A_i$ and $B_i$ are chains with 3 taxa. The construction works for every $k \geq 3$.}
\label{fig:tight}
\end{center}
\end{figure}
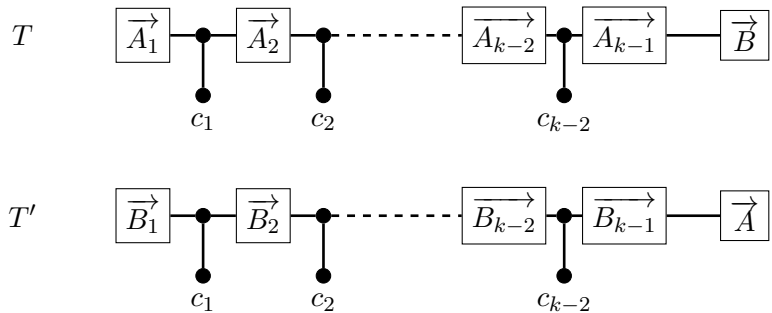

\textbf{Remark}. We observe that the size of the kernel for the same reduction rules can alternatively be bounded to $7k + O(1)$ by leveraging
the \emph{generator} approach of \cite{KelkL18}. The details of the
generator machinery used there are beyond this article but the high-level idea is as follows; we ignore additive terms here. If the uMAF has $k$ components, this can be modelled as adding leaves to a cubic multigraph with $3k$ edges, and distributing $2k$ \emph{breakpoints} (i.e. $k$ per tree) across those edges, such that each edge receives 0, 1 or 2 breakpoints. An edge with 0 breakpoints can receive at most 3 taxa because otherwise the chain reduction could be applied. Now, we observe that as soon as an edge with 1 breakpoint receives strictly more than 2 taxa, a cherry is induced in one of the input trees. In a similar vein, as soon as an edge with 2 breakpoints receives strictly more than 1 taxon, a cherry is induced in one (or both) of the input trees, or contradicts the assumed optimality of the uMAF.

Crucially, when the input consists of 2 caterpillars, there are only 4 cherries to divide across the edges: 2 per input tree. Hence, there are at most a \emph{constant} number of edges of the following types: (i) 1 breakpoint and more than 2 taxa, (ii) 2 breakpoints and more than 1 taxon. The counting equation is then optimized (i.e. describing the worst case: two caterpillars with the highest number of leaves possible) by taking $k$ 0-breakpoint edges with 3 taxa on each, and $2k$ edges each with 2 taxa and 1 breakpoint. This gives $7k + O(1)$ taxa.

\section{An improved FPT branching algorithm for caterpillars}\label{Section Branching Alg}

In this section we prove the following theorem. {As usual, $O^{*}$ notation suppresses polynomial factors.

\begin{Theorem} \label{thm:cat25}
Let $T$ and $T'$ be caterpillars on the same set of taxa $X$. For each $k$, it can be determined
in time $O^{*}(2.49^k)$ whether $T$ and $T'$ have an agreement forest with at most $k$ components.
\end{Theorem}

We start with some simple observations. Recall that we vacuously allow trees on 3 fewer
taxa to be regarded as caterpillars.

\begin{Observation} \label{Cor all trees are caterpillars}
Let $T$ be a caterpillar on $X$ and let $X' \subseteq X$. Then $T|X'$ is also a caterpillar.
\end{Observation}

\begin{Lemma} \label{Lemma if a then l or r is 1}
Let $T$ and $T'$ both be caterpillars on $X$. Suppose $a \in X$ is part of a cherry in $T$. For an agreement forest $A$ which does not contain $\{a\}$ as a singleton component, let $B_a$ be the component of $A$ that contains $a$. Let  $B_a = \{a\} \cup L \cup R$ where  where $L$ and $R$ are the taxa in the two subtrees sibling to $a$ in $T|B_a = T'|B_a$ (see Fig. \ref{Fig intro L and R}). Then $|L| \leq 1$ or $|R| \leq 1$. In particular: $a$ is part of a cherry in $T|B_a = T'|B_a$.
\end{Lemma}

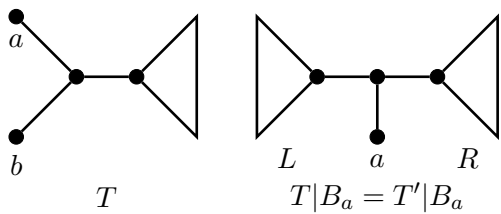
\begin{figure}
\begin{center}
\begin{tikzpicture}[scale = 0.8]
	
    \node (F1) at (1.5,-2) {$T$};
    \vertex [label=below:$a$] (a) at (0,1) {};
    \vertex [label=below:$b$] (b) at (0,-1) {};
    \vertex (xab) at (1,0) {};
    \vertex (x) at (2,0){};
    \draw [line width = 1pt]
    (xab) edge (a)
    (xab) edge (b)
    (xab) edge (x)
    (x) -- (3,1) -- (3,-1) -- (x);

    \node (F1) at (6,-2) {$T|B_a = T'|B_a$};
    \node (BL) at (4.5,-1.35) {$L$};
    \node (BR) at (7.5,-1.35) {$R$};

    \vertex [label=below:$a$] (a) at (6,-1) {};
    \vertex (xa) at (6,0) {};
    \vertex (xl) at (5,0){};
    \vertex (xr) at (7,0){};
    \draw [line width = 1pt]
    (xa) edge (a)
    (xa) edge (xl)
    (xa) edge (xr)
    (xl) -- (4,1) -- (4,-1) -- (xl)
    (xr) -- (8,1) -- (8,-1) -- (xr);

\end{tikzpicture}
\caption{If $a$ is part of a cherry in $T$, then an agreement forest of $T$ and $T'$ in which $a$ is part of a non-singleton component $B_a$ has the property that $|L| \leq 1$ or $|R| \leq 1$.}
\label{Fig intro L and R}
\end{center}
\end{figure}

\begin{proof}
Towards a contradiction,  suppose $|L| \geq 2$ and $|R|\geq 2$. Then $a$ is not part of a cherry in $T|B_a = T'|B_a$. However, let $c$ be the taxon from $L \cup R$ that is closest to $a$ in $T$. Due to the caterpillar structure of $T$, it follows that $\{a,c\}$ is a cherry in $T|B_a$, yielding a contradiction. Hence, $|L| \leq 1$ or $|R| \leq 1$. Combining this with the fact that $L \cup R \neq \emptyset$ (due to the assumption that $B_a \neq \{a\}$), we have that $a$ is part of a cherry in $T|B_a = T'|B_a$.
\end{proof}

\subsection{High-level idea of the branching algorithm}

We draw inspiration from the branching algorithm of Chen et al. \cite{ChenFS15} and earlier
work in a similar vein such as by Whidden et al. \cite{WhiddenBZ13}. We start with two caterpillar trees $T$ and $T'$ on $X$. The high-level idea is to progressively cut edges in one of the caterpillars, say $T'$, to obtain a forest $F'$ (initially $F'=T'$) with an increasing number of components, until it becomes an \emph{agreement} forest for $T$ and $T'$. Each edge cut increases the size of the forest by 1. Hence, if we wish to know whether there exists an agreement forest with at most $k$ components, we can make at most $k-1$ edge cuts in $T'$. As in earlier work, at each step we apply various ``tidying up'' steps:
\begin{enumerate}
\item If a singleton component is created in $F'$ comprising a single taxon $a$, then we delete $a$ from both $T$ and $F'$;
\item Degree-2 nodes are always suppressed;
\item Common cherries are always reduced into a single taxon in both trees.
\end{enumerate}
The way we will choose edges to cut, combined with the tidying up steps, ensures that (unlike $T'$) $T$ remains connected at every step. More formally: if, at a given step, $T'$ has been cut into a forest $F'$ and $X'$ is the union of the taxa in $F'$, then $T$ will have been transformed into $T|X'$. To keep notation light we will henceforth refer simply to $T$ and $F'$, with the understanding that we are actually referring to the tree-forest pair encountered at a specific iteration of the algorithm.

The general decision problem is as follows. We are given a $(T,F')$ tree-forest pair on a set $X' \subseteq X$ of taxa, and a parameter $k$, and we wish to know: is it possible to transform $F'$ into an agreement forest making at most $k-1$ cuts (more formally: an agreement forest for $T|X'$ and $T'|X'$)? If we can answer this question we will, in particular, be able to answer the original question of whether the original input  $(T, T')$ has an agreement forest with at most $k$ components.
We use a recursive branching algorithm to answer this query and use $T(k)$ to denote the running time required to answer the query. 

We will show that $T(k)$ is $O^{*}(2.49^k)$.  We will make heavy use of the following fact. \begin{Observation}
\label{obs:alwayscat}
At every step of the algorithm, $T$ will always be a caterpillar and, due to Observation \ref{Cor all trees are caterpillars}, $F'$ will always be a forest of caterpillars.
\end{Observation}

Our first branching rule is well-known in the literature; see also Figure \ref{Fig general starting point branching alg differnte trees}.\\
\\
\textbf{Branching rule 0} (\cite{WhiddenBZ13}). \emph{Suppose that $T$ contains a cherry $\{a,b\}$ and that, in $F'$, $a$ and $b$ are in different trees $T'_a$ and $T'_b$. Then in any agreement forest obtained by cutting edges in $F'$,
at least one of $a$ and $b$ is a singleton component. Hence we can branch by guessing whether to delete $a$, or to delete $b$.}\\
\\
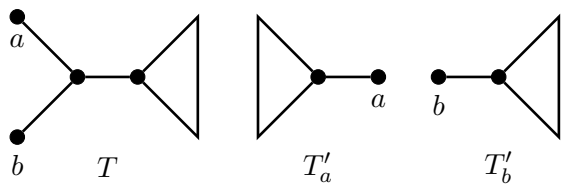
\begin{figure}[!h]
\begin{center}
\begin{tikzpicture}[scale = 0.8]
	
    \node (F1) at (1.5,-1.5) {$T$};
    \vertex [label=below:$a$] (a) at (0,1) {};
    \vertex [label=below:$b$] (b) at (0,-1) {};
    \vertex (xab) at (1,0) {};
    \vertex (x) at (2,0){};
    \draw [line width = 1pt]
    (xab) edge (a)
    (xab) edge (b)
    (xab) edge (x)
    (x) -- (3,1) -- (3,-1) -- (x);

    \node (BL) at (5,-1.5) {$T'_a$};
    \node (BR) at (8,-1.5) {$T'_b$};

    \vertex (l) at (5,0){};

    \vertex [label=below:$a$] (a) at (6,0) {};
    \vertex [label=below:$b$] (b) at (7,0) {};

    \vertex (r) at (8,0){};
    \draw [line width = 1pt]
    (l) edge (a)
    (b) edge (r)
    (l) -- (4,1) -- (4,-1) -- (l)
    (r) -- (9,1) -- (9,-1) -- (r);	
\end{tikzpicture}
\caption{Branching rule 0: taxa $a$ and $b$ are a cherry in $T$ but are in different trees $T'_a$ and $T'_b$ of $F'$. Any agreement forest reached from this point must contain at least one of $a$ and $b$ as a singleton component.}
\label{Fig general starting point branching alg differnte trees}
\end{center}
\end{figure}

This immediately yields the recurrence $T(k) \leq 2T(k-1)$, where the $2T(k-1)$ term corresponds to the fact that we delete $a$ or $b$, and each such guess requires one edge cut. Such branching yields a bound of at most $O^*(2^k)$ and is thus clearly compatible with our overall goal of $O^*(2.49^k)$. So we can safely apply Branching Rule 0 whenever we come across such a situation.

We henceforth assume Branching Rule 0 does not apply. This means that $a$ and $b$ are a cherry in $T$ and are part of the \emph{same} tree $T'_{ab}$ in $F'$, as shown in Figure \ref{Fig general starting point branching alg same trees}. We can also assume that $a$ and $b$ do not have a common parent in $T'_{ab}$, because then $\{a,b\}$ would be a common cherry and would have been reduced by the ``tidying up'' steps. Due to Observation \ref{obs:alwayscat}, $T$ and $T'_{ab}$ will always be caterpillars. This is the starting point for all our new branching rules.

\begin{figure}[h]
\begin{center}
\begin{tikzpicture}[scale = 0.8]
	
    \node (F1) at (1.5,-2) {$T$};
    \vertex [label=below:$a$] (a) at (0,1) {};
    \vertex [label=below:$b$] (b) at (0,-1) {};
    \vertex (xab) at (1,0) {};
    \vertex (x) at (2,0){};
    \draw [line width = 1pt]
    (xab) edge (a)
    (xab) edge (b)
    (xab) edge (x)
    (x) -- (3,1) -- (3,-1) -- (x);

    \node (F1) at (6.5,-2) {$T'_{ab}$};

    \node (BL) at (4.5,-1.5) {$L$};
    \node (BR) at (8.5,-1.5) {$R$};

    \vertex (l) at (5,0){};

    \vertex (pa) at (6,0) {};
    \vertex [label=below:$a$] (a) at (6,-1) {};
    \vertex (pb) at (7,0) {};
    \vertex [label=below:$b$] (b) at (7,-1) {};

    \vertex (r) at (8,0){};
    \draw [line width = 1pt]
    (pa) edge (a)
    (pb) edge (b)
    (l) edge (pa)
    (pa) edge [dashed] (pb)
    (pb) edge (r)
    (l) -- (4,1) -- (4,-1) -- (l)
    (r) -- (9,1) -- (9,-1) -- (r);

\end{tikzpicture}
\caption{The new branching rules all consider the situation that $a$ and $b$ form a cherry in $T$, and are part of the same tree $T'_{ab}$ in $F'$ where they do not form a cherry. Note that the parents of $a$ and $b$ might not be adjacent, and that one or both of $L, R$ might be empty.}
\label{Fig general starting point branching alg same trees}
\end{center}
\end{figure}
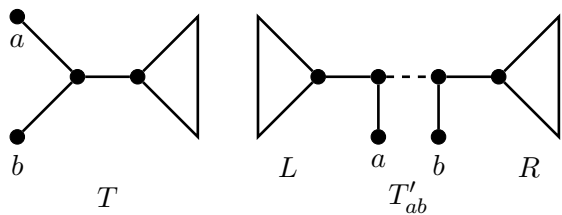

\subsection{Branching Rule 1: The parents of taxa $a$ and $b$ \emph{are} adjacent in $T'_{ab}$ (Bound: $2^k$)} \label{Sec Branching rule 1} 

In this case there are two branches. When we write ``Cut off $L$'' we mean: cut the edge $e_{La}$ between the $L$ subtree and the parent of $a$, as shown in Figure \ref{Fig general starting point branching alg same trees}. ``Cut off $R$'' means: cut the edge $e_{Rb}$ between the $R$ subtree and the parent of $b$. The two branches are:

\begin{enumerate}
    \item Cut off $L$ in $T'_{ab}$ 
    \item Cut off $R$ in $T'_{ab}$ 
\end{enumerate}
This yields at most $T(k) \leq 2T(k-1)$, and thus $T(k) \leq 2^k$. (Note that if $L$ or $R$ is empty, then $\{a,b\}$ form a common cherry and would have been reduced.)

To prove the correctness of this branching rule we will show that, if $F'$ can be transformed into an agreement forest $A$ by cutting at most $k-1$  additional edges, then there exists an agreement forest $A'$, obtained by making at most $k-1$ additional cuts to $F'$ such that the following holds: there exists an edge  $e \in \{ e_{La}, e_{Rb} \}$ such that, for every component $B \in A'$, $T[B]$ does not contain the edge $e$.

\begin{proof}

Let  $A$ be an agreement forest obtained from $F'$ with at most $k-1$ additional cuts.
First, suppose that $a$ is not a singleton in $A$. Then, $a$ is part of a component $B_a \in A$ such that $a$ is in a cherry with some taxon $c$ in $T|B_a = T'|B_a$ (Lemma \ref{Lemma if a then l or r is 1}). If $c=b$ then it
is not possible for $B_a$ to simultaneously contain $a, b$, an element from $L$ and an element of $R$: because $T|B_a$ would contain $\{a,b\}$ as a cherry, but $T'|B_a$ would not.   Hence, at least one of $e_{La}, e_{Rb}$ is not used by $B_a$, and no other component of $A$ uses that edge either (because $B_a$ contains, in addition to $a$, at least one other taxon).

\newpage

So, suppose that $c \neq b$. We distinguish two subcases.

\begin{enumerate}
    \item[(i)]$b \in B_a$. Then $B_a = \{a,b,c\}$. This is because $T|B_a = T'|B_a$ contains $\{a,c\}$ as a cherry (by assumption) but also $\{a,b\}$ as a cherry (due to the topology of $T$) and a taxon can only be in two or more cherries of a phylogenetic tree if the tree has three taxa. Once again, this means that it is safe to cut at least one of $e_{La}, e_{Rb}$.
    \item[(ii)] $b \not \in B_a$.  Note that $b$ is then necessarily a singleton component in $A$ due to $\{a,b\}$ being a cherry in $T$ and $c \in B_a$. Now, we claim that deleting $B_a$ and $\{b\}$ from $A$, and replacing them with $(B_a \setminus \{c\} \cup \{b\})$ and $\{c\}$, yields an agreement forest $A'$ (with the same number of components as $A$). This main reason for this is that in $A$, due to Lemma \ref{Lemma if a then l or r is 1}, $|B_a \cap L| \leq 1$ or $|B_a \cap R| \leq 1$. In particular: if $B_a \cap L = \{c\}$, or $B_a \cap R = \{c\}$, then it is easy to verify that $A'$ is still an agreement forest. Alternatively, suppose that $B_a \cap L$ contains $c$ and at least one other taxon. (A symmetrical analysis holds if $B_a \cap R$ contains $c$ and at least one other taxon.) Irrespective of how close $c$ is to $a$ in $B_a$, we have that $|B_a \cap R|=0$, because otherwise $a$ and $c$ do not form a cherry. It can then again be easily checked that, irrespective of the size of $B_a$ that $A'$ is an agreement forest.
\end{enumerate}
Crucially, there exists $e \in \{e_{La}, e_{Rb}\}$ such that all components in $A'$ avoid $e$.\\

At this stage of the proof we have shown that, if $a$ is not a singleton in $A$, then the branching rule is safe. Due to symmetry between $a$ and $b$, the same proof shows that, if $b$ is not a singleton in $A$, then the branching rule is safe. Hence, the only situation left to consider is when both $a$ and $b$ are singletons in $A$. We produce a new agreement forest $A'$ by removing $\{a\}$ and $\{b\}$, introducing the new component $\{a,b\}$ and (if necessary) splitting at most component in $A$; this is the at most one component of $A$ that uses the edge between the parents of $a$ and $b$ in $T'_{ab}$. Clearly $|A'| \leq |A|$. Moreover, neither of the two edges $e_{La}, e_{Rb}$ are used by $A'$. So the branching rule is safe.
    
\end{proof}

\subsection{Branching Rule 2: There is at least one taxon between $a$ and $b$ in $T'_{ab}$}

Assume $a$ and $b$ have a number of taxa between them in $T'_{ab}$. Denote this chain of taxa between $a$ and $b$ as $C$ with $i \geq 1$ being the number of taxa in $C$, as shown in Figure \ref{Fig ab same tree with chain}. We note that one or both of $L, R$ can potentially be empty. If $L$ is empty then $a$ is in a cherry in $T'_{ab}$ with some taxon $c \in C$. The taxon $c$ can then be moved out of $C$ to take the role of $L$, reducing the number of taxa in $C$ by 1. The same transformation can be applied if $R$ is empty. This allows us to assume that both $L$ and $R$ are non-empty, but at the price of reducing the length of $C$ by at most 2. Observe that if both $L$ and $R$ are empty then, before moving any taxa out of $C$, we have $i \geq 3$ because otherwise we will be in Branching Rule 1 after moving taxa from $C$ into $L$ and $R$.

In this branching rule we deal with the case $i \geq 2$ first.

\begin{figure}[h]
\begin{center}
\begin{tikzpicture}[scale = 0.8]
	
    \node (F1) at (1.5,-2) {$T$};
    \vertex [label=below:$a$] (a) at (0,1) {};
    \vertex [label=below:$b$] (b) at (0,-1) {};
    \vertex (xab) at (1,0) {};
    \vertex (x) at (2,0){};
    \draw [line width = 1pt]
    (xab) edge (a)
    (xab) edge (b)
    (xab) edge (x)
    (x) -- (3,1) -- (3,-1) -- (x);

    \node (F1) at (7.5,-2) {$T'_{ab}$};
    \node (C) at (7.5,-0.75) {$C$};

    \node (BL) at (4.5,-1.5) {$L$};
    \node (BR) at (10.5,-1.5) {$R$};

    \vertex (l) at (5,0){};

    \vertex (pa) at (6,0) {};
    \vertex [label=below:$a$] (a) at (6,-1) {};
    \vertex (pb) at (9,0) {};
    \vertex [label=below:$b$] (b) at (9,-1) {};

    \vertex (r) at (10,0){};
    
    \draw [line width = 1pt]
    (pa) edge (a)
    (pb) edge (b)
    (l) edge (pa)
    (pa) edge [dashed] (pb)
    (pb) edge (r)
    (l) -- (4,1) -- (4,-1) -- (l)
    (r) -- (11,1) -- (11,-1) -- (r);

\end{tikzpicture}
\caption{In the tree $T'_{ab}$ taxa $a$ and $b$ have a chain $C$ of taxa between them. Note that it is also possible that one or both of $L,R$ are empty but by moving taxa out of the chain $C$ we can assume that $L$ and $R$ are both non-empty.}
\label{Fig ab same tree with chain}
\end{center}
\end{figure}
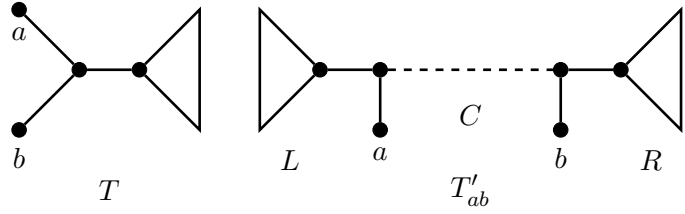

\subsubsection{Branching Rule 2.1: There are at least two taxa between $a$ and $b$ in $T'_{ab}$ (Bound: $2.49^k$, which is runtime-defining)}\label{Sec C branching}

After making sure both $L$ and $R$ are non-empty we assume that $C$ still contains at least 2 taxa. If $C$ only contains 1 taxon we will need to use the later branching rule.

Let $A$ be an agreement forest obtained from $F'$ by at most $k-1$ additional cuts. Either $a$ is a singleton component in $A$; or $b$ is a singleton component in $A$; or $a$ and $b$ are together in some component. If $a$ and $b$ are together in some component $B_{ab} \in A$, then the following implications hold:

\begin{itemize}
    \item If $B_{ab}$ contains at least one taxon from $L$, then it does not contain any taxa from $R$ or $C$.
    \item If $B_{ab}$ contains at least one taxon from $R$, then it does not contain any taxa from $L$ or $C$.
    \item If $B_{ab}$ contains some taxon from $C$, then it contains exactly one such taxon, and no taxa from $L$ or $R$ i.e. $|B_{ab}|=3$.
\end{itemize}
Note that, if $B_{ab}$ contains exactly one taxon $d \in C$, then all the taxa outside $B_{ab}$ whose parents lie on the path from cherry $\{a,b\}$ to $d$ in $T$, will necessarily be singleton components in $A$. Hence, if $d$ is not the taxon in $C$ closest to cherry $\{a,b\}$ in $T$, we can replace the components $B_{ab} = \{a,b,d\}$ and $\{c\}$ with the components $\{a,b,c\}$and $\{d\}$, where $c$ is the taxon in $C$ that is closest to the cherry $\{a,b\}$ in $T$. 

Thus, we have 5 branches:

\begin{enumerate}
    \item Cut off only $b$ 
    \item Cut off only $a$ 
    \item Cut off $R$ and all taxa in $C$ 
    \item Cut off $L$ and all taxa in $C$ 
    \item Cut off $L$ and $R$ and all taxa in $C$ except $c$ (which is closest to $\{a,b\}$ in $T$) 
\end{enumerate}
This gives the recurrence $T(k) = 2T(k-1)+3T(k-i-1)$. Given that $i \geq 2$, this is less than or equal to $2T(k-1) + 3T(k-3)$. It can be verified that the solution to the induced recurrence is the positive root of the polynomial $x^3 - 2x^2 - 3 = 0$, which is
$x \approx 2.48558...$. Hence, $T(k) \leq 2.49^k$. Note that this determines the bottleneck for the entire branching algorithm.

\subsubsection{Branching Rule 2.2: There is exactly one taxon $c$ between $a$ and $b$ in $T'_{ab}$, and $c$ adjacent to cherry $\{a,b\}$ in $T$ (Bound: $1.62^k$)}

See Figure \ref{fig:162}. Here $L$ and $R$ are definitely both non-empty, because otherwise we would be in Branching Rule 1 (i.e. $a$ and $b$ would have adjacent parents in $T'_{ab}$ because $a$ and $c$ or $b$ and $c$ will share the same parent.). The branches in this case are:

\begin{enumerate}
    \item Cut off $L$ and $R$ 
    \item Cut off $c$ 
\end{enumerate}

\begin{figure}[h]
\begin{center}
\begin{tikzpicture}[scale = 0.8]
	    
    \node (F2) at (2,-2) {$T$};
    \vertex [label=below:$a$] (a) at (0,1) {};
    \vertex [label=below:$b$] (b) at (0,-1) {};
    \vertex [label=below:$c$] (c) at (2,-1) {};
    \vertex (xab) at (1,0) {};
    \vertex (xc) at (2,0) {};
    \vertex (t) at (3,0) {};
    \draw [line width = 1pt]
    (xab) edge (a)
    (xab) edge (b)
    (xc) edge (c)
    (xab) edge (t)
    (t) -- (4,1) -- (4,-1) -- (t);
    
    \node (F1) at (8,-2) {$T'_{ab}$};
    \node (BL) at (5.5,-1.35) {$L$};
    \node (BR) at (10.5,-1.35) {$R$};
    \vertex [label=below:$a$] (a) at (7,-1) {};
    \vertex [label=below:$b$] (b) at (9,-1) {};
    \vertex [label=below:$c$] (c) at (8,-1) {};
    \vertex (xa) at (7,0) {};
    \vertex (xc) at (8,0) {};
    \vertex (xb) at (9,0) {};
    \vertex (xl) at (6,0){};
    \vertex (xr) at (10,0){};
    \draw [line width = 1pt]
    (xa) edge (a)
    (xb) edge (b)
    (xc) edge (c)
    (xl) edge (xr)
    (xl) -- (5,1) -- (5,-1) -- (xl)
    (xr) -- (11,1) -- (11,-1) -- (xr);

\end{tikzpicture}
\caption{Here there is only one taxon $c$ between $a$ and $b$ in $T'_{ab}$, and $c$ is adjacent to cherry $\{a,b\}$ in $T$. We can either cut $c$ off, or both $L$ and $R$.}
\label{fig:162}
\end{center}
\end{figure}
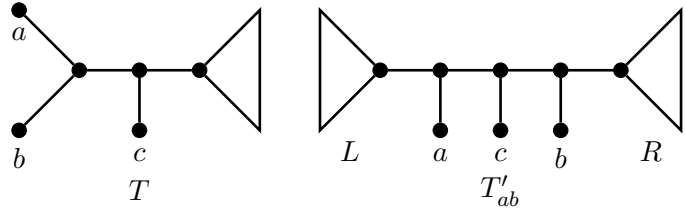

\begin{proof}
    Proof of correctness: Let $A$ be an agreement forest obtained by making at most $k-1$ additional cuts in $F'$. Observe that if there is a component $B \in A$ such that $\{a,b,c\} \subseteq B$, then $B=\{a,b,c\}$, and hence we can safely cut off $L$ and $R$. This situation is modelled by the first branch.

Suppose that there is no such component. Let $B$ now be the component in $A$ that contains $c$. The situation that $c$ is a singleton component is modelled by the second branch. So, suppose that $B$ contains $c$ and at least one other taxon. If neither $a$ nor $b$ is in $B$, then $a$ and $b$ are necessarily singleton components. It can then easily be verified that due to the topology of $T$ and $T'_{ab}$ replacing $B$ and $\{a\}$ with $(B \setminus \{c\}) \cup \{a\}$ and $\{c\}$ yields a new agreement forest $A'$ of size equal to $A$; $c$ is now a singleton component and the second branch can be used. If $B$ contains (without loss of generality, due to symmetry) $a$ but not $b$, then $b$ is a singleton and $B$ does not contain taxa from both $L$ and $R$. Here replacing $B$ and $\{b\}$ with
$(B \setminus \{c\}) \cup \{b\}$ and $\{c\}$ again yields a new agreement forest
with size equal to $A$ where $c$ is a singleton again. Making it fully safe to only use these two branches.
\end{proof}

This gives the recurrence $T(k)= T(k-1) + T(k-2)$. It can be verified that the solution to the induced recurrence is the positive root of the polynomial $x^2 - x - 1 = 0$ with $x \approx 1.6180$. 
Hence $T(k) \leq 1.62^k$.

\subsubsection{Branching Rule 2.3: Exactly one taxon $c$ between $a$ and $b$ in $T'_{ab}$, but $c$ is not adjacent to cherry $\{a,b\}$ in $T$ (Bound: $2.45^k$)}

See Figure \ref{fig:Yintheway}. There are a number of different instances we can come across. We will group them together into two main cases. Note that we can again assume that $L$ and $R$ are both not empty. Remember that if either $L$ or $R$ is empty we end up in an instance for Branching Rule 1. 

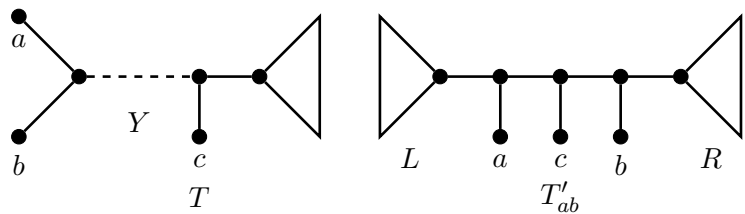
\begin{figure}[h]
\begin{center}
\begin{tikzpicture}[scale = 0.8]
	    
    \node (F2) at (2,-2) {$T$};
    \node (Y) at (1, -0.75) {$Y$};
    \vertex [label=below:$a$] (a) at (-1,1) {};
    \vertex [label=below:$b$] (b) at (-1,-1) {};
    \vertex [label=below:$c$] (c) at (2,-1) {};
    \vertex (xab) at (0,0) {};
    \vertex (xc) at (2,0) {};
    \vertex (t) at (3,0) {};
    \draw [line width = 1pt]
    (xab) edge (a)
    (xab) edge (b)
    (xc) edge (c)
    (xab) edge [dashed] (xc)
    (xc) edge (t)
    (t) -- (4,1) -- (4,-1) -- (t);
    
    \node (F1) at (8,-2) {$T'_{ab}$};
    \node (BL) at (5.5,-1.35) {$L$};
    \node (BR) at (10.5,-1.35) {$R$};
    \vertex [label=below:$a$] (a) at (7,-1) {};
    \vertex [label=below:$b$] (b) at (9,-1) {};
    \vertex [label=below:$c$] (c) at (8,-1) {};
    \vertex (xa) at (7,0) {};
    \vertex (xc) at (8,0) {};
    \vertex (xb) at (9,0) {};
    \vertex (xl) at (6,0){};
    \vertex (xr) at (10,0){};
    \draw [line width = 1pt]
    (xa) edge (a)
    (xb) edge (b)
    (xc) edge (c)
    (xl) edge (xr)
    (xl) -- (5,1) -- (5,-1) -- (xl)
    (xr) -- (11,1) -- (11,-1) -- (xr);

\end{tikzpicture}
\caption{Here there is a non-empty set of taxa
 $Y$ between cherry $\{a,b\}$ and $c$ in $T$, where $c$ is the sole taxon between $a$ and $b$ in $T'_{ab}$.}
\label{fig:Yintheway}
\end{center}
\end{figure}

First assume that $Y$ has no common taxa with $T'_{ab}$, $Y \cap T'_{ab}=\emptyset$. With this assumption it is sufficient to branch off into three branches. For each branch we denote (an upper bound on) the magnitude of the parameter in the corresponding recursive call.

\begin{enumerate}
    \item Cut off all of the taxa in $Y$ $\Rightarrow$ $(k-1)$ (because $Y$ contains at least one taxon)
    \item Cut off $a$ and $b$ $\Rightarrow$ $(k-2)$
    \item Cut off $L$, $R$, and $c$ $\Rightarrow$ $(k-3)$.
\end{enumerate}
This gives the recursion:
\[
    T(k)\leq T(k-1) + T(k-2) + T(k-3)
\]
It can be verified that the solution to the induced recurrence is the positive root of the polynomial $x^3 -x^2 - x - 1 = 0$ with $x \approx 1.8393$.
Hence $T(k) \leq 1.84^k$.

\begin{proof}
    Proof of correctness:
        Let $A$ be an agreement forest obtained by making at most $k-1$ additional cuts in $F'$.\
        
        If $A$ has a component $B$ such that ${a,b} \in B$ but $B\neq \{a,b\}$ all taxa in $Y$ have to be singletons, so branch 1 applies.\
        
        If $B=\{a,b\}$ than $c$ must be a singleton, and none of the taxa in $L$ can be connected to the taxa in $R$, so branch 3 applies.\

        If there is a component $B$ such that $|B| \geq 2$, and exactly one of $a, b \in B$, then we are back in branch 1, since all the taxa in $Y$ must be singletons.\
        
        Finally, if both $a$ and $b$ are singletons, we are in branch 2.

\end{proof}

Now assume $Y$ has taxa that are also present in $T_{ab}'$, $Y \cap T'\neq\emptyset$. Define $D=Y\cap L$ and $E = Y\cap R$. For $d\in D$ define $L_{d\rightarrow}$ as the set of taxa in $L$ whose parents are on the path from $d$ to $a$ in $T'_{ab}$. Define $L_{\leftarrow d}$ as the set of taxa in $L$ left of $d$. Thus $L = L_{\leftarrow d} \cup \{d\} \cup L_{d\rightarrow}$ (imagine reading $T_{ab}'$ from the cherry on the left towards $a$). Define the same for $e\in E$ but mirrored to get $R = \ R_{\leftarrow e} \cup \{e\} \cup R_{e\rightarrow}$ (this time read $T_{ab}'$ from $b$ towards the cherry on the right). See Figure \ref{fig:YinthewaynonemptyDE} for an illustration.

\begin{figure}[h]
\begin{center}
\begin{tikzpicture}[scale = 0.8]
	    
    \node (F2) at (2,-3) {$T$};
    \node (Y) at (1.5, -0.75) {$Y$};
    \vertex [label=below:$a$] (a) at (-1,1) {};
    \vertex [label=below:$b$] (b) at (-1,-1) {};
    \vertex [label=below:$c$] (c) at (3,-1) {};
    \vertex [label=above:$d$] (d) at (1,1) {};
    \vertex [label=above:$e$] (e) at (2,1) {};

    \vertex (xab) at (0,0) {};
    \vertex (xc) at (3,0) {};
    \vertex (xd) at (1,0) {};
    \vertex (xe) at (2,0) {};
    \vertex (t) at (4,0) {};
    \draw [line width = 1pt]
    (xab) edge (a)
    (xab) edge (b)
    (xc) edge (c)
    (xd) edge (d)
    (xe) edge (e)
    (xab) edge [dashed] (xc)
    (xc) edge (t)
    (t) -- (5,1) -- (5,-1) -- (t);
    
    \node (F1) at (11,-3) {$T'_{ab}$};
    \node (BL) at (8,-2) {$L$};
    \node (BL) at (7,-1) {$L_{\leftarrow d}$};
    \node (BL) at (9,-1) {$L_{d\rightarrow}$};
    \node (BR) at (14,-2) {$R$};
    \node (BR) at (13,-1) {$R_{\leftarrow e}$};
    \node (BR) at (15,-1) {$R_{e\rightarrow}$};
    \vertex [label=below:$a$] (a) at (10,-1) {};
    \vertex [label=below:$b$] (b) at (12,-1) {};
    \vertex [label=below:$c$] (c) at (11,-1) {};
    \vertex [label=above:$d$] (d) at (8,1) {};
    \vertex [label=above:$e$] (e) at (14,1) {};
    \vertex (xa) at (10,0) {};
    \vertex (xc) at (11,0) {};
    \vertex (xb) at (12,0) {};
    \vertex (xd) at (8,0) {};
    \vertex (xe) at (14,0) {};
    \vertex (xl) at (7,0){};
    \vertex (xr) at (15,0){};
    \draw [line width = 1pt]
    (xa) edge (a)
    (xb) edge (b)
    (xc) edge (c)
    (xa) edge (xb)
    (xd) edge (d)
    (xe) edge (e)
    (xl) edge [dashed] (xa)
    (xr) edge [dashed] (xb)
    (xl) -- (6,1) -- (6,-1) -- (xl)
    (xr) -- (16,1) -- (16,-1) -- (xr);

\end{tikzpicture}
\caption{Illustrating the definition $L = L_{\leftarrow d} \cup \{d\} \cup L_{d\rightarrow}$ and $R = \ R_{\leftarrow e} \cup \{e\} \cup R_{e\rightarrow}$.}
\label{fig:YinthewaynonemptyDE}
\end{center}
\end{figure}
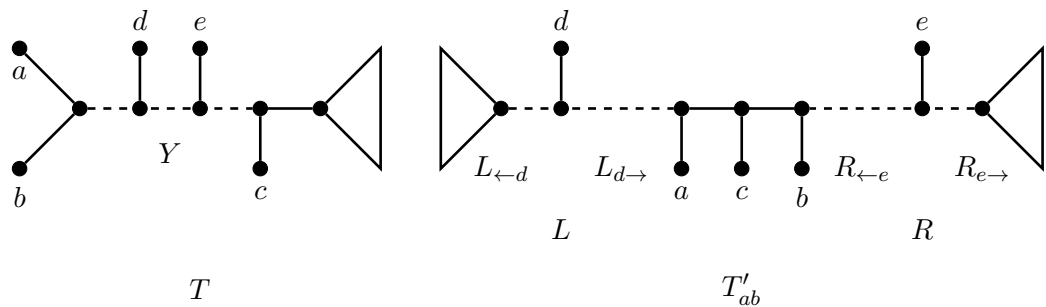

We create the following branches:

\begin{enumerate}
    \item Cut off $a$ and $b$ 
    \item Cut off $L$ and $c$ 
    \item Cut off $c$ and $R$ 
    \item Cut off $L$ and $R$
    \newpage
    \item $\forall d\in D$ 
        \begin{itemize}
            \item Cut off $b$ 
            \item Cut off all taxa in $D\setminus\{d\}$
            \item Cut off $L_{\leftarrow d}$ (with a single cut) if this set is non-empty
            \item Cut off all taxa in $L_{d\rightarrow}$ if this set is non-empty
        \end{itemize}
    \item $\forall e\in E$ 
        \begin{itemize}
            \item Cut off $a$ 
            \item Cut off all taxa in $E\setminus\{e\}$
            \item Cut off all taxa in $R_{\leftarrow e}$ if this set is non-empty
            \item Cut $R_{e\rightarrow}$ (with a single cut) if this set is non-empty 
        \end{itemize}
\end{enumerate}

Before we show the bound on the recursion we will prove that these branches are correct.

\begin{proof}
    Let $A$ be an agreement forest obtained by making at most $k-1$ additional cuts in $F'$. If $A$ has components $\{a\}$ and $\{b\}$ branch 1 applies.\
    
    Say $A$ does not contain both $\{a\}$ and $\{b\}$ but has a component $B$ containing both $a$ and $b$. Then if $B$ contains more taxa they must either all come from $L$, or all come from $R$, or $B = \{a,b,c\}$. In all those cases branches 2-4 applies.\
    
    Say $a$ is in $B\neq \{a\}$ but $b$ is not. Then $\{b\} \in A$. Taxon $a$ must still form a cherry with another taxon $t \neq b$ \rubenAugust{in $B$}. This taxon $t$ can be one of five different taxa:
    \begin{enumerate}
        \item a taxon in $(T-Y) \cap L$
        \item a taxon in $(T-Y) \cap R$
        \item a taxon in $D$
        \item a taxon in $E$
        \item $t$ is taxon $c$
    \end{enumerate}
    
    Case 1 : In this case because $t$ is not in $Y$ we know that if $A$ does not contain $\{c\}$ it must contain $\{a,c,t\}$.
    If that is the case we can change $A$ by replacing $\{a,c,t\}$ and $\{b\}$ with $\{a,b,c\}$ and $\{t\}$. So branch 4 applies.\ 
    
    From this point assume that $\{c\}\in A$. Suppose $B= \{a,t\}$; then we can change $A$ by replacing $B$ and $\{b\}$ with $(B\setminus \{t\}) \cup\{b\}$ and $\{t\}$ so branches 2-4 applies (because we are back in the earlier situation when $a$ and $b$ are together in a component).\
    
    Assume from this point that $B\neq\{a,t\}$. Say $B$ contains more taxa from $L$. Then it cannot have any taxa from $R$. We can then change $A$ by replacing $B$ and $\{b\}$ with $(B\setminus \{t\}) \cup\{b\}$ and $\{t\}$ so branch 3 applies.\

    Say $B$ contains taxa from $R$, meaning $B$ cannot have other taxa from $L$ in it besides $t$. We adjust $A$ by replacing $B$ and $\{b\}$ with $(B\setminus \{t\}) \cup\{b\}$ and $\{t\}$ so branch 2 applies.\\
    
    Case 2: Analogous to case 1.\\

    Case 3: Suppose  $B=\{a,t\}$. 

    If $A$ does not contain $\{c\}$ then we can cut $c$ out of its component (increasing the number of components by one) but then compensate by adding $b$ to $B$, reducing the number of components by one, so branch 3 applies. If $A$ does contain $\{c\}$ then we can add $b$ to $B$ (reducing the number of components by one) so branch 3 applies again.

    Say $B$ contains more taxa from $L$, this can include taxa from $D$, than just $t$. In all cases we adjust $A$ by replacing $B$ and $\{b\}$ with $(B\setminus \{t\}) \cup\{b\}$ and $\{t\}$. So branch 3 applies. If $B=\{a,c,t\}$ than we can do the same adjustment to $A$ but this time branch 4 applies.
    
    Finally, say $B$ contains taxa from $R$ and $c$, meaning $B$ cannot have other taxa from $L$ in it besides $t$. However unlike previous cases branches 1-4 do not apply. This time branch 5 applies. This is because even though $a$ and $t$ form a cherry in $B$, $B$ also contains taxon $c$ and possibly taxa from $E$. At first glance one might think we can adjust $A$ like we did before, but we can't.\\

    Case 4: Analogous to case 1 we can adjust $A$ always in such a way that we end up in branch 2,3 or 4.\\

    Case 5: If $B$ has exactly 2 or 3 elements we have already seen how we can change $A$ by replacing $B$ and $\{b\}$ with $(B\setminus \{t\}) \cup\{b\}$ and $\{t\}$ so branch 2 or 3 applies.\

    Say $B$ contains 4 or more taxa. Because $a$ forms a cherry with $c$ we know $B$ can't have any more taxa from $D$ or $E$. So it can only have taxa from $L\setminus D$ or $R \setminus E$. It cannot have both. If it did $a$ would not form a cherry with $t=c$ in $T_{ab}'$. This means that in both cases we change $A$ again like we did many times before, by replacing $B$ and $\{b\}$ with $(B\setminus \{t\}) \cup\{b\}$ and $\{t\}$. So branch 2 or 3 applies again. And we are done.\\

    To finish the proof we state that we can use the same arguments starting from the assumption that $B$ does contain $b$ but not $a$, and $b$ forms a cherry with a taxon labeled $t$. In all cases we either adjust $A$ ending up in branches 2-4, or we have $t\in E$ and $B$ containing $\{b,t, c\}$ and possibly more from $L$, especially $D$, so branch 6 applies.
\end{proof}

Now we move on to the bound of the recursion.
For branch 5 we can, for each $d \in D$, get a pessimistic upper bound on the magnitude of the recursion as follows: minus $1$ for cutting off $b$, minus $0$ for cutting off $L_{\leftarrow d}$ or all taxa $L_{d\rightarrow}$ this is because they can be empty, and finally minus $(|D|-1)$ due to cutting off $|D|-1$ taxa in in $D\setminus \{d\}$, giving a total magnitude of $k-1-(|D|-1)=k-|D|$. The same count can be done for branch 6.

\newpage

This gives the recursion:
\[
    T(k) \leq  4T(k-2) + |D|T(k-|D|) + |E|T(k-|E|)
\]

This will become too large when the size of $D$ or $E$ is one. So in order to get around this we explicitly consider boundary cases. First, however, we deal with the general case. Here we will assume that if $D$ is nonempty than $L$ contains at least 2 taxa, and that if $E$ is nonempty than $R$ also contains at least 2 taxa. Because of the assumption on $L$, $R$ and $D$, $E$, we can conclude that} at least one of $L_{\leftarrow d}$ \rubenAugust{and}  $L_{d\rightarrow}$ is nonempty for each $d \in D$ and at least one of  $R_{\leftarrow e}$ \rubenAugust{and} $R_{e\rightarrow}$ is nonempty for each $e \in E$. 

With these assumptions we get another pessimistic upper bound of branch 5 on the magnitude of the recursion but this one is safe: minus $1$ for cutting off $b$, minus $1$ for cutting off at least one taxon in $L_{\leftarrow d}$ or $L_{d\rightarrow}$ now that we have the assumption at least one of these two sets is non-empty, and finally minus $(|D|-1)$ due to cutting off $|D|-1$ taxa in in $D\setminus \{d\}$, giving a total magnitude of $k-1-1-(|D|-1)=k-1-|D|$. The same count can be done for branch 6. This gives the recursion:
\[
    T(k) \leq  4T(k-2) + |D|T(k-1-|D|) + |E|T(k-1-|E|)
\]

It can be proven that the solution to the induced recurrence is the positive root of the polynomial $f_{d,e}(x) = x^{d+1} - 4x^{d-1} - ex^{d-e}-d$ where $d=|D|$ and $e=|E|$. Here, without loss of generality we assume $0\leq e\leq d$ and $d\neq0$. We can assume that $e \leq d$ because we can if necessary always swap the positions of $a$ and $b$, due to $a$ and $b$ being a cherry in $T$. We can assume $d \neq 0$ because (prior to any relabelling) we know that at least one of $D, E$ is non-empty because $Y$ has taxa that are also present in $T_{ab}'$. Crucially the largest positive real root of $f_{d,e}(x)$ is no larger than $2.45$. Hence $T(k) \leq 2.45^k$. To prove this we look at two properties $f_{d,e}(x)$ that make proving the desired bound 
straightforward.

\begin{Observation} \label{Obs f is positive}
    For all integers $d,e$ such that $0\leq e \leq d$ and $d\neq0$ we have that if $x > \sqrt{6}$ then $f_{d,e}(x) > 0$.
\end{Observation}
\begin{proof}
    Proof by induction on $d$ and $e$. The induction will be done in 3 phases. First we will prove that the statement is true for all integers $d$ with $e=0$. Then we will prove the statement for the special case $d=e$. We end by showing that the statement is true for all integers $e$ between $0$ and strictly less than $d$.\\

    First phase:
    Base cases $e=0$ and $d=1$.
    \[
        f_{1,0}(x) = x^2-5 
    \]
    Clearly if $x > \sqrt{6}$ then $f_{1,0}(x) > 0$. Now we move on towards $d+1$ with $d\geq 1$.
    \begin{align*}
            f_{d+1,0}(x) & = x^{d+2} - 4x^{d} -d-1 \\
            & = x(x^{d+1} - 4x^{d-1} -d) + dx-d-1        
    \end{align*}

    Using the induction hypothesis we can assume $x(x^{d+1} - 4x^{d-1} - d)> 0 $. Leaving only $dx-d-1$. Observe that $x>\sqrt{6}>2$ making $dx-d-1> 2d-d-1=d-1$. We also know that $d\geq 1$ thus $dx-d-1>0$ and we are done.

    Now we move to the second phase: the special case $e=d$:
    \[
        f_{d,d}(x) = x^{d+1} - 4x^{d-1} - 2d
    \]
    First the base $e=d=1$:
    \[
        f_{1,1}(x)=x^2-6
    \]
    Clearly the statement is true for $f_{1,1}(x)$ so let's move on towards $d+1$:
    
    \begin{align*}
            f_{d+1,d+1}(x) & = x^{d+2} - 4x^{d} -2d-2 \\
            & = x(x^{d+1} - 4x^{d-1} -2d) + 2(dx-d-1)        
    \end{align*}
    Using the induction hypothesis we can again assume $x(x^{d+1} - 4x^{d-1} - 2d)> 0 $. Leaving only $2(dx-d-1)$ for which we can repeat the arguments proving $dx-d-1>0$ and we are done.

    We finish with induction on $e$ for values between $0$ and $d$. We have already proven that the statement is true for $f_{d,0}$ for all values of $d$. Now assume it to be true for $f_{d,e}$ with $1 \leq e < d$. We will prove the statement is true for $f_{d,e+1}$. If $e+1=d$ we are done because of the second phase. So we may assume $e+1 < d$.
    \[
        f_{d, e+1} = x^{d+1} - 4x^{d-1} - (e+1)x^{d-e-1}-d
    \]
    Observe that $e+1 \leq 2e < xe$ because $x > 2$ allowing the following steps:

    \begin{align*}
        f_{d, e+1}(x) & = x^{d+1} - 4x^{d-1} - (e+1)x^{d-e-1}-d \\
         & \geq x^{d+1} - 4x^{d-1} - xex^{d-e-1}-d \\
         & = x^{d+1} - 4x^{d-1} - ex^{d-e}-d\\
         & = f_{d,e}(x) > 0
    \end{align*}
    And we are done.
\end{proof}

\begin{Observation} \label{Obs f is negative}
    For all $d,e$ such that $0\leq e \leq d$ and $d\neq0$ we have that if $x\in [0,2]$ then $f_{d,e}(x) < 0$
\end{Observation}
\begin{proof}
    This is just a straightforward algebra argument.
\[
        x^{d+1} - 4x^{d-1} - ex^{d-e}-d = x^{d-1}(x^2-4) - ex^{d-e}-d 
\]
    $(x^2-4) \leq 0$ for all values of $x\in [0,2]$. The rest of the formula is obviously less than zero.
\end{proof}

With these two observations we will now prove that our last set of branches are nicely bounded. 

\begin{proof}
    From Observation \ref{Obs f is positive} and \ref{Obs f is negative} we can conclude that for each value of $d,e$ there exists a $x_{0}\in [2, \sqrt{6}]$ such that $x_0$ is a root of $f_{d,e}(x)$ and for all $x>x_0$ we get $f_{d,e}(x)>0$.
\end{proof}

In conclusion the recursion $T(k) \leq 4T(k-2) + |D|T(k-1-|D|) + |E|T(k-1-|E|)$ is bounded by $(\sqrt{6})^k < 2.45^k$.\\

However as mentioned before we still have some boundary cases to consider. In each of these cases we have that at least $L=D=\{d\}$ or $R=E=\{e\}$ holds. This way branch 5 or 6 will only make 1 cut, namely cutting off $b$ and $a$ respectively. This creates a $T(k-1)$ term in the recursion and with already four $T(k-2)$ terms we would exceed the desired $2.49^k$ bound.

Say $L=D=\{d\}$, we make no assumptions on $Y$ or $R$ other than the fact that both are non-empty. Branch 4 cuts off $L$ and $R$ allowing the possibility that an agreement forest found in this branch contains $B=\{a,b,c\}$. However if this is the case, with the same kind of argument we used to reduce the number of branches in section \ref{Sec C branching}, we can adjust $A$ by replacing $B$ and $\{d\}$ with $(B\setminus \{c\})\cup \{d\}$ and $\{c\}$, making branch 4 useless. This only works because $L$ contains a single element.

If $E$ is empty we get the following branches:

\begin{enumerate}
    \item Cut off $a$ and $b$ 
    \item Cut off $L$ and $c$ 
    \item Cut off $c$ and $R$ 
    \item[5.] Cut off $b$
\end{enumerate}

This will give the recursion $T(k)= T(k-1) + 3 T(k-2)$ and we are done, as this solves to less than $2.40^k$.

So now let's assume $E$ is not empty. To get rid of the $T(k-1)$ created by branch 5 in the recursion we will look at the type of agreement forests that can be produced by branch 5. If $d$ is a singleton we are done, because now if $a$ is also a singleton we are in branch 1: recall that branch 5 always cuts off $b$ so any forest reached via branch 5 will have $b$ as a singleton. If $a$ is not a singleton we can always adjust the agreement forest such that $a$ and $b$ are in the same component and branch 2 or 3 applies.

So we can move on to cases where $d$ forms a cherry with some other taxon. $d$ can form a cherry with a taxon in $R$ or $a$ or $c$. If $d$ does not form a cherry with $a$ branch 1 applies because branch 5 always cuts $b$. So we may assume $d$ forms a cherry with $a$. If $B=\{a,d\}$ than we are done, because we can adjust $A$ by replacing $B$, $\{b\}$ and $B_c$, where $B_c$ is the component containing $c$, with $B\cup \{b\}$, $B_c\setminus\{c\}$ and $\{c\}$. So branch 3 applies. If $a,d \in B$ only contains taxa in $R$ other than $a$ and $d$ we can adjust $A$ by replacing $B$ and $\{b\}$ with $B\setminus\{d\} \cup \{b\}$ and $\{d\}$ so branch 2 applies. And finally if $B=\{a,d,c\}$, we can change that to $\{a,b,d\}$ so branch 3 applies.

The only case where we cannot use previously mentioned branches is if $a,d \in  B$ contains $c$ and at least one taxon from $E$,
or when it contains $c$ and at least one taxon from $R\setminus E$. To cover these situations, we modify branching rule 5 to take advantage of this extra structure.

In the first case $B$ can only contain a single taxon from $E$ that must be between $c$ and $d$ in $T$. It cannot contain more than one taxon from $E$. If it did, $B$ would have a different topology in both caterpillars. So we can cut the rest of $R\setminus E$ by cutting $R_{e\rightarrow}$ and all taxa in $R_{e\leftarrow}$. We may also cut all the taxa in $E\setminus\{e\}$. In order for this to work in our favor with respect to lowering the bound on the recursion created by this adjusted set of branches we must assume that $E\setminus\{e\}$ or $R\setminus E$ is non-empty. (The case that they both are, meaning $R=E=\{e\}$, will be the final case we discuss). However we don't know what the best option is for $B$ with respect to all possible taxa in $E$ which are between $d$ and $c$. So for each of them we make a branch.

Now the second case, $a,d\in B$ only contains taxa in both $R\setminus E$ and $\{c\}$. In this case $B$ cannot contain a taxon from $E$ because of the topology of $T$, meaning we can cut all taxa in $E$. Notice that in both cases we end up cutting all taxa in $E$ that are between $a$ and $d$ in $T$. 
So in conclusion if $L=D=\{d\}$ we cut all taxa in $E$ between $d$ and $a$. If there exist none there must be at least one taxon in $E$ between $d$ and $c$ in $T$ because we assumed $E$ to be non-empty. For each of those we create a single branch.

\begin{enumerate}
    \item[5.1] If there is a taxon $e \in E$ between $d$ and $c$ in $T$ than for all such $e$:
        \begin{itemize}        
            \item Cut off $b$ 
            \item Cut off all taxa in $R_{\leftarrow e}$ if this set is non-empty
            \item Cut $R_{e\rightarrow}$ (with a single cut) if this set is non-empty 
            \item Cut the rest of $E\setminus \{e\}$ if this set is nonempty
        \end{itemize}
    \item[] Else:
        \begin{itemize}
            \item Cut off $b$
            \item Cut off all taxa in $E$
        \end{itemize}
\end{enumerate}

Notice that if $E$ is empty we get the exact same set of branches as on the previous page. Branch 5.1 will always cut at least twice: once for cutting off $b$ and once for cutting at least one taxon $e\in E$. If $|E|=1$ we thus still get $T(k-2)$, in both the top and bottom case. In the top case we create $|E|$ number of branches each cutting $b$ and at least one taxon off from $R_{e\rightarrow}$, $R_{e\leftarrow}$ or $E\setminus\{e\}$ for any size of $E$ by assumption.  If $|E|\geq2$ the worst case is when $R=E$ (worst meaning the least number of cuts done in each branch). Then we get $|E|T(k-|E|)$. Notice we get the following: $|E|T(k-|E|) \leq T(k-2) + (|E|-1)T(k-|E|)$. This formula also works for the case when $|E|=1$ because we just saw that we are guaranteed to get at least $T(k-2)$. Now we get the general formula for all the branches:
\[
    T(k) = 3* T(k-2)+ T(k-2) + (|E|-1)T(k-|E|) + |E|T(k-1-|E|).
\]

Now we substitute $d+1 = e$ to get $f_{e,d}(x)$, the same formula as before but now with $d$ and $e$ switched. So we can conclude that this recursion is bounded by $2.49^k$. The same analysis can be done when $R=E=\{e\}$ and $D\setminus\{d\}$ or $L\setminus D$ is non-empty. All that remains is the case $L=D=\{d\}$ and $R=E=\{e\}$.

Say $L=D=\{d\}$ and $R=E=\{e\}$, then there are two cases, $d$ left and $e$ right in $T$ and the other way around shown in  Figure \ref{fig:optionsLR1}. Consider the left case. An agreement forest $A$ requires at least two components to cover the taxa $\{a,b,c,d,e\}$, and each of these components is a subset of $\{a,b,c,d,e\}$ (because $T'_{ab}$ only contains the taxa $a,b,c,d,e$). We can therefore replace these components in the agreement forest with $\{a,d,c,e\}$ and $\{b\}$. In other words, we can deterministically conclude that it is safe to cut $b$ off. Similarly, in the right case we can replace the two or more components covering $\{a,b,c,d,e\}$ with $\{b,e,c,d\}$ and $\{a\}$, allowing us to deterministically conclude that we can cut $a$ off. In both cases it is clear that the recursion is bounded by $2.49^k$ and we are done.

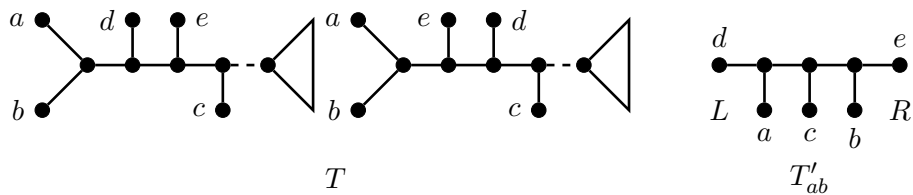
\begin{figure}[!h]
\begin{center}
\begin{tikzpicture}[scale = 0.6]

    \node (F2) at (5.5,-5.5) {$T$};
    \vertex [label=left:$a$] (a) at (-1,-2) {};
    \vertex [label=left:$b$] (b) at (-1,-4) {};
    \vertex [label=left:$c$] (c) at (3,-4) {};
    \vertex [label=left:$d$] (d) at (1,-2) {};
    \vertex [label=right:$e$] (e) at (2,-2) {};
    \vertex (xab) at (0,-3) {};
    \vertex (xc) at (3,-3) {};
    \vertex (xd) at (1,-3) {};
    \vertex (xe) at (2,-3) {};
    \vertex (t) at (4,-3) {};
    \draw [line width = 1pt]
    (xab) edge (a)
    (xab) edge (b)
    (xc) edge (c)
    (xd) edge (d)
    (xe) edge (e)
    (xab) edge  (xc)
    (xc) edge [dashed] (t)
    (t) -- (5,-2) -- (5,-4) -- (t);
    
    \vertex [label=left:$a$] (a) at (6,-2) {};
    \vertex [label=left:$b$] (b) at (6,-4) {};
    \vertex [label=left:$c$] (c) at (10,-4) {};
    \vertex [label=right:$d$] (d) at (9,-2) {};
    \vertex [label=left:$e$] (e) at (8,-2) {};
    \vertex (xab) at (7,-3) {};
    \vertex (xc) at (10,-3) {};
    \vertex (xd) at (9,-3) {};
    \vertex (xe) at (8,-3) {};
    \vertex (t) at (11,-3) {};
    \draw [line width = 1pt]
    (xab) edge (a)
    (xab) edge (b)
    (xc) edge (c)
    (xd) edge (d)
    (xe) edge (e)
    (xab) edge  (xc)
    (xc) edge [dashed] (t)
    (t) -- (12,-2) -- (12,-4) -- (t);
    
    \node (F1) at (16,-5.5) {$T'_{ab}$};
    \node (BL) at (14,-4) {$L$};
    \node (BR) at (18,-4) {$R$};
    \vertex [label=below:$a$] (a) at (15,-4) {};
    \vertex [label=below:$b$] (b) at (17,-4) {};
    \vertex [label=below:$c$] (c) at (16,-4) {};
    \vertex [label=above:$d$] (d) at (14,-3) {};
    \vertex [label=above:$e$] (e) at (18,0-3) {};
    \vertex (xa) at (15,-3) {};
    \vertex (xc) at (16,-3) {};
    \vertex (xb) at (17,-3) {};
    \draw [line width = 1pt]
    (xa) edge (a)
    (xa) edge (d)
    (xb) edge (b)
    (xb) edge (e)
    (xc) edge (c)
    (xa) edge (xb);
    
\end{tikzpicture}
\caption{The two different cases if $|L|=|D|=|R|=|E|=1$.}
\label{fig:optionsLR1}
\end{center}
\end{figure}

\newpage

\section{Discussion and conclusions} \label{sec:discussion}

\subsection{Caterpillars and TBR distance: a complex relationship}
We recall the following definition of a \emph{tree bisection and reconnection} move, defined on unrooted binary phylogenetic trees, and its corresponding distance.
 Let $T$ be a phylogenetic tree on $X$. Apply the following three-step operation to $T$:
\begin{enumerate}
    \item Delete an edge in $T$ and suppress any resulting degree-2 vertex. Let $T_1$ and $T_2$ be the two resulting phylogenetic trees.
    \item If $T_1$ (resp. $T_2$) has at least one edge, subdivide an edge in $T_1$ (resp. $T_2$) with a new vertex $v_1$ (resp. $v_2$) and otherwise set $v_1$ (resp. $v_2$) to be the single isolated vertex of $T_1$ (resp. $T_2$).
    \item Add a new edge $\{v_1,v_2\}$ to obtain a new phylogenetic tree $T'$ on $X$.
\end{enumerate}
We say that $T'$ has been obtained from $T$ by a single  {\it tree bisection and reconnection (TBR) operation} (or, {\it TBR move}). We define the TBR {\it distance}  between two phylogenetic trees $T$ and $T'$ on $X$,  denoted by $d_{TBR}(T,T')$, to be the minimum number of TBR operations that are required to transform $T$ into $T'$. As mentioned earlier it is well known that, on unrooted binary trees the TBR distance between two trees is equal to the size of an uMAF, minus one \cite{AllenS01}.

In the main part of this article we did not give a formal definition of TBR distance, focussing only on agreement forests. The reason for this, is that when focussing on a restricted subset of tree topologies as we do here (caterpillars), the definition of TBR distance becomes more complex and consequently so does the relationship with agreement forests. In particular: when defining the TBR distance between two caterpillars, should (i) all the intermediate trees \emph{also} be caterpillars, or (ii) is it permitted that the intermediate trees be general trees? In (ii) the equivalence between TBR distance and agreement forests remains. However, in (i) the relationship with agreement forests breaks down somewhat.  In particular, it is possible that although two caterpillars have a maximum agreement forest with $k$ components the intermediate trees constructed by the $k-1$ TBR moves aren't all caterpillars.

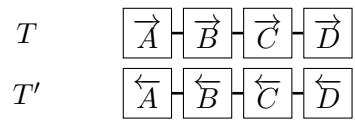
\begin{figure}[!h]
\begin{center}
\begin{tikzpicture}[scale = 0.8]    
    \node (T1) at (2, 0) {$T$};
    \node (A) [draw, minimum width=0.5cm, minimum height=0.5cm] at (4, 0) {$\overrightarrow{A}$};
    \node (B) [draw, minimum width=0.5cm, minimum height=0.5cm] at (5, 0) {$\overrightarrow{B}$};
    \node (C) [draw, minimum width=0.5cm, minimum height=0.5cm] at (6, 0) {$\overrightarrow{C}$};
    \node (D) [draw, minimum width=0.5cm, minimum height=0.5cm] at (7, 0) {$\overrightarrow{D}$};
    
    \draw [line width = 1pt]
    (A) edge (B)
    (B) edge (C)
    (C) edge (D);
    
    \node (T1) at (2, -1) {$T'$};
    \node (A) [draw, minimum width=0.5cm, minimum height=0.5cm] at (4, -1) {$\overleftarrow{A}$};
    \node (B) [draw, minimum width=0.5cm, minimum height=0.5cm] at (5, -1) {$\overleftarrow{B}$};
    \node (C) [draw, minimum width=0.5cm, minimum height=0.5cm] at (6, -1) {$\overleftarrow{C}$};
    \node (D) [draw, minimum width=0.5cm, minimum height=0.5cm] at (7, -1) {$\overleftarrow{D}$};
    
    \draw [line width = 1pt]
    (A) edge (B)
    (B) edge (C)
    (C) edge (D);
	
\end{tikzpicture}
\caption{Caterpillars $T$ and $T'$ are made of blocks of chains, oriented in opposing directions in each tree. Each chain contains 3 taxa.}
\label{Fig example TBR not equal to uMAF-1}
\end{center}
\end{figure}

For example the two caterpillars in Figure \ref{Fig example TBR not equal to uMAF-1} have a maximum agreement forest of size 4, $\{A,B, C, D\},$ and we can thus obtain $T'$ from $T$ after 3 TBR moves. However, the only way to obtain $T'$ from $T$ with 3 moves is that the first intermediate tree is \emph{not} a caterpillar; if we restrict to caterpillars, 4 or more moves are required. It would be interesting to elucidate this relationship further, and whether there is a variant of agreement forests that models this variant of TBR distance on caterpillars.

\subsection{Future research}

A number of interesting questions remain. We have shown that computation of uMAF on caterpillars remains hard; what kind of topological restrictions on input trees make uMAF easy? Can we develop new reduction rules which, for caterpillars, reduce the kernel bound below $7k$? Similarly, what kind of new branching rules would be required to reduce the running time of the caterpillar branching algorithm below $2.49^k$? Can the insights from our $2.49^k$ branching algorithm be leveraged to improve the current state-of-the-art $3^k$ branching algorithm for general trees? Finally, we echo the point made in \cite{kelk2022deep} and elsewhere: can the analysis of branching rules, and reduction rules, be systematized somehow?

\section{Acknowledgements}

 We thank Simone Linz and Steve Chaplick for useful discussions. Ruben Meuwese was supported by the Dutch Research Council (NWO) KLEIN 1 grant \emph{Deep kernelization for phylogenetic discordance}, project number OCENW.KLEIN.305.
 
\bibliography{Bibliografie}{}
\bibliographystyle{plain}
	
\end{document}